\documentclass[a4paper]{scrartcl}
\usepackage{amssymb}
\usepackage{amsmath}
\usepackage{amsthm}
\usepackage{latexsym}
\usepackage{enumerate}
\usepackage{graphicx,color}
\usepackage{dsfont}
\usepackage{float}
\usepackage[textwidth=17mm]{todonotes}
\usepackage{placeins}
\usepackage[normalem]{ulem}

\DeclareFontFamily{OT1}{pzc}{}
\DeclareFontShape{OT1}{pzc}{m}{it}{<-> s * [1.10] pzcmi7t}{}
\DeclareMathAlphabet{\mathpzc}{OT1}{pzc}{m}{it}

\usepackage{xcolor}
\usepackage{enumitem}

\usepackage[round,authoryear]{natbib}
\newcommand*{\doi}[1]{\href{http://dx.doi.org/#1}{doi: #1}}
\usepackage[autostyle=true]{csquotes} % Required to generate language-dependent quotes in the bibliography

\makeatletter
\newcommand*\bigcdot{\mathpalette\bigcdot@{.5}}
\newcommand*\bigcdot@[2]{\mathbin{\vcenter{\hbox{\scalebox{#2}{$\m@th#1\bullet$}}}}}
\makeatother

\colorlet{shadecolor}{gray!60}
\numberwithin{equation}{section}

\numberwithin{figure}{section}

\theoremstyle{plain}
\newtheorem{theorem}{Theorem}[section]
\newtheorem{proposition}[theorem]{Proposition}

\theoremstyle{definition}
\newtheorem{definition}[theorem]{Definition}

\newtheorem{example}[theorem]{Example}
\newtheorem{remark}[theorem]{Remark}

\usepackage{tikz}
\usetikzlibrary{arrows,positioning,calc}
\usetikzlibrary{decorations.pathreplacing}
\tikzset{
    >=stealth',
    punkt/.style={
           rectangle,
           rounded corners,
           draw=black, thick,
           text width=5.5em,
           minimum height=2em,
           text centered},
    punktl/.style={
           rectangle,
           rounded corners,
           draw=black, thick,
           text width=7em,
           minimum height=2em,
           text centered},
    punktll/.style={
           rectangle,
           rounded corners,
           draw=black, thick,
           text width=9em,
           minimum height=2em,
           text centered},           
    pil/.style={
           ->,
           shorten <=4pt,
       shorten >=4pt
    },
    pildotted/.style={
           ->,
           shorten <=4pt,
           shorten >=4pt,
  dotted,
  },
    punktf/.style={
           rectangle,
           text width=4.0em,
           minimum height=1.5em,
           text centered},
    punktfTop/.style={
           rectangle,
           text width=4.0em,
           minimum height=1.5em,
           text centered,
           append after command={
               [thick,shorten >=0.2bp, shorten <=0.2bp]
               (\tikzlastnode.north west)edge(\tikzlastnode.north east)
}
    },
    punktfBot/.style={
           rectangle,
           text width=4.0em,
           minimum height=1.5em,
           text centered,
           append after command={
               [thick,shorten >=0.2bp, shorten <=0.2bp]
               (\tikzlastnode.south west)edge(\tikzlastnode.south east)
            }
    }
}

\newcommand{\indep}{\perp \!\!\! \perp}

\colorlet{lightred}{red!20}
\colorlet{lightgray}{gray!25}
\definecolor{platinum}{RGB}{229,228,226}
\definecolor{ash}{RGB}{178,190,181}
\definecolor{lightblue}{HTML}{b3cde0}

\newcommand\xqed[1]{%
  \leavevmode\unskip\penalty9999 \hbox{}\nobreak\hfill
  \quad\hbox{#1}}
\newcommand\demormk{\xqed{$\triangledown$}}
\newcommand\demoex{\xqed{$\circ$}}

\newcommand\demodef{\xqed{$\triangle$}}

\makeatletter
\newcommand{\pushright}[1]{\ifmeasuring@#1\else\omit\hfill$\displaystyle#1$\fi\ignorespaces}
\newcommand{\pushleft}[1]{\ifmeasuring@#1\else\omit$\displaystyle#1$\hfill\fi\ignorespaces}
\makeatother

\usepackage{hyperref}

\makeatletter
\DeclareFontFamily{OMX}{MnSymbolE}{}
\DeclareSymbolFont{MnLargeSymbols}{OMX}{MnSymbolE}{m}{n}
\SetSymbolFont{MnLargeSymbols}{bold}{OMX}{MnSymbolE}{b}{n}
\DeclareFontShape{OMX}{MnSymbolE}{m}{n}{
    <-6>  MnSymbolE5
   <6-7>  MnSymbolE6
   <7-8>  MnSymbolE7
   <8-9>  MnSymbolE8
   <9-10> MnSymbolE9
  <10-12> MnSymbolE10
  <12->   MnSymbolE12
}{}
\DeclareFontShape{OMX}{MnSymbolE}{b}{n}{
    <-6>  MnSymbolE-Bold5
   <6-7>  MnSymbolE-Bold6
   <7-8>  MnSymbolE-Bold7
   <8-9>  MnSymbolE-Bold8
   <9-10> MnSymbolE-Bold9
  <10-12> MnSymbolE-Bold10
  <12->   MnSymbolE-Bold12
}{}

\let\llangle\@undefined
\let\rrangle\@undefined
\DeclareMathDelimiter{\llangle}{\mathopen}%
                     {MnLargeSymbols}{'164}{MnLargeSymbols}{'164}
\DeclareMathDelimiter{\rrangle}{\mathclose}%
                     {MnLargeSymbols}{'171}{MnLargeSymbols}{'171}
\makeatother

\newcommand*\diff{\mathop{}\!\mathrm{d}}

\newcommand{\stkout}[1]{\ifmmode\text{\sout{\ensuremath{#1}}}\else\sout{#1}\fi}

\usepackage{authblk}

\title{Transaction time models in multi-state life insurance} 

\author[1]{Kristian Buchardt}
\author[2]{Christian Furrer}
\author[1,2,$\star$]{Oliver Lunding Sandqvist}

\affil[1]{\footnotesize PFA Pension, Sundkrogsgade 4, DK-2100 Copenhagen \O, Denmark.}
\affil[2]{\footnotesize Department of Mathematical Sciences, University of Copenhagen, Universitetsparken 5, DK-2100 Copenhagen \O, Denmark.}
\affil[$\star$]{\footnotesize Corresponding author. E-mail: \href{mailto:oliver.s@math.ku.dk}{oliver.s@math.ku.dk}.}

\date{\vspace{-8mm}}

\begin{document}

\maketitle

\begin{abstract}
In life insurance contracts, benefits and premiums are typically paid contingent on the biometric state of the insured. Due to delays between the occurrence, reporting, and settlement of changes to the biometric state, the state process is not fully observable in real-time. This fact implies that the classic multi-state models for the biometric state of the insured are not able to describe the development of the policy in real-time, which encompasses handling of incurred-but-not-reported and reported-but-not-settled claims. We give a fundamental treatment of the problem in the setting of continuous-time multi-state life insurance by introducing a new class of models: transaction time models. The relation between the transaction time model and the classic model is studied and a result linking the present values in the two models is derived. The results and their practical implications are illustrated for disability coverages, where we obtain explicit expressions for the transaction time reserve in specific models.
\end{abstract}

\vspace{5mm}


\vspace{5mm}


\vspace{2mm}


\vspace{5mm}

\section{Introduction} \label{sec:Introduction}

The payments stipulated in life insurance contracts are usually an agreement on what payments are to be made for different possible outcomes of the biometric state of the insured (e.g.\ whether the insured is active, disabled, dead, etc.). For this reason, multi-state life insurance models take modeling of the biometric state of the insured as their starting point. The multi-state approach to life insurance dates back to at least~\citet{Hoem1969}. Here, the prospective reserve is defined as the discounted probability-weighted future payments, which, as noted in~\citet{Norberg1991}, corresponds to the expected present value of future payments given the information generated by the biometric state process. The introduction of an underlying stochastic process generating the payments introduces structure to the problem of predicting the cash flow at future points in time, due to the temporal dependencies of the process. This added structure of the payments is not in itself an assumption when the payments stipulated in the insurance contracts are formulated in terms of the biometric state of the insured. It is rather a way to introduce more a priori knowledge about the workings of the product into the mathematical model. All other things being equal, this makes the models more powerful.

Consequently, multi-state modeling seems a natural approach to modeling life insurance products. However, in the multi-state modeling literature, one also often assumes that the biometric state process generating the payments equals the process that generates the available information, see e.g.~\citet{Norberg1991}, ~\citet{BuchardtMollerSchmidt2015}, \citet{Djehiche2016}, ~\citet{Bladt2020}, and~\citet{Christiansen2021} to name a few. This is rarely the case, since information about changes to the biometric state can be delayed or erroneous. A simple example of this phenomenon is the delay that occurs when an insured becomes disabled; it might take some time for the insured to report the event to the insurer. Between the occurrence of the disability and the time of reporting, the claim is an IBNR (Incurred-But-Not-Reported) claim. As long as the insured has not reported the disability, the insurer will continue believing that the insured is active.  Hence, the information that the insurer has is different from the full information about the biometric state of the insured. To describe this phenomenon in more detail, and discuss how to approach reserving under the insurer's available information, we introduce the concepts of \textit{valid time} and \textit{transaction time} in the next section: Essentially, the valid time of an event is the time that it occurs, while the transaction time is the time that the event is registered in the insurers records.

It turns out that these concepts are also useful in clarifying the similarities and differences between life and non-life insurance products as well as between the models employed in the respective fields. The fact that payments in life insurance are deterministic functions of the biometric state process makes it so one does not have to estimate a separate distribution for the payment sizes; once the distribution of the state process is specified, the distribution of the payments follows. This is not the case in non-life insurance, and one therefore resorts to modeling the distribution of the observed payments directly. However, as will be explained, the biometric state process is a valid time object, while the observed payments are transaction time objects. This fact leads to key differences in the life and non-life insurance models. One such key difference is that it is more straightforward to formulate IBNR and RBNS (Reported-But-Not-Settled) models in non-life insurance, as one can construct the RBNS model entirely in transaction time. The IBNR model may subsequently be constructed in two steps. First, one models the occurrence times of claims, which are valid time objects, as well as the corresponding reporting delays. Second, one leverages the RBNS model to find the expected cash flow of a claim conditional on the time of occurrence and the reporting delay. In life insurance models, one has to link the valid time payments to the transaction time concepts of IBNR and RBNS, and it is not obvious how to do this.

 Our main contributions are: the introduction of the \textit{basic bi-temporal structure assumptions} defined in Section~\ref{sec:TTM}, the derivation of Theorem~\ref{theorem:PZequalsPYR} that links the present values in valid and transaction time, and an application of this theorem, namely Example~\ref{ex:SimpleReserving}. The first of these contributions establishes an explicit link between transaction and valid time processes. The second utilizes this link to obtain a tractable relation between the present values in valid and transaction time. To further obtain a tractable relation between the valid and transaction time reserves, more structure on how the transaction time information affects the distribution of the valid time process needs to be imposed. This is exactly what is explored in a simple example involving RBNS claims, culminating in Example~\ref{ex:SimpleReserving}, which constitutes the third main contribution. The example is kept simple for illustrative purposes, but our general framework also allows for the study of intricate examples that provide a more complete picture of IBNR and RBNS reserving. Such applications are the raison d'être of the framework.

The paper is structured as follows. In Section~\ref{sec:VTTT}, the terms valid time and transaction time are given more precise definitions and discussed in the context of life insurance. An overview of the use of valid time and transaction time information in the insurance literature is provided, and similarities and differences between the situation in life insurance and non-life insurance are made explicit. The section ends by defining the class of piecewise deterministic processes, which constitute the basic building blocks for our model constructions. Section~\ref{sec:VTM} constructs a model for the insurance contract in valid time similarly to how the classic life insurance multi-state models are constructed. Section~\ref{sec:TTM} introduces the novel concept of a transaction time model corresponding to a valid time model, and this transaction time model is constructed. In Section~\ref{sec:Reserving}, the valid and transaction time reserves are defined, and a result relating the transaction time present value to the valid time present value is derived. In a model for disability insurance where coverage depends on the origin of the disability, we show how this relation can be utilized to obtain a relation between the corresponding reserves. Finally, the dynamics of the transaction time reserve is derived and discussed. 

\section{Valid and transaction time} \label{sec:VTTT}

We now introduce the terms valid time and transaction time. These concepts are used to describe data that arises from a time-varying process. We outline how these types of data are currently being used in the life and non-life insurance literature. Subsequently, we introduce a class of stochastic processes which we use to model processes generating valid time and transaction time data.

\subsubsection*{Valid and transaction time data} \label{subsubsec:ValidTransactData}

The terms valid time and transaction time provide a natural terminology for describing information that is registered with delays and uncertainty. Valid time and transaction time are concepts stemming from the design of databases, specifically temporal databases, where time-varying information is recorded. The valid/transaction time taxonomy was developed in~\citet{Snodgrass1985}. There, valid time is defined as the time that an event occurs in reality, while transaction time is defined as the time when the data concerning the event was stored in the database. Hence, valid time is concerned with when events occurred (historical information), while transaction time is concerned with when events were observed (rollback information).

As noted by~\citet{Snodgrass1985}, an important difference between valid time and transaction time are the types of information updates that are permitted. A transaction time may be added to the database, but is never allowed to be changed after the fact due to the forward motion of time. In contrast, a valid time is always subject to change, since discrepancies between the history as it actually occurred and the representation of the history as stored in the database will often be detected after the fact. The authors argue that both valid time and transaction time are needed to fully capture time-varying behavior.

A database that contains both valid time and transaction time is called a \textit{bi-temporal} database. Such a database supplies both historical and rollback information. Historical information e.g.\ \textit{``Where was Taylor employed during 2010?"} is supplied by valid time, while rollback information e.g.\ \textit{``In 2010, where did the database believe Taylor was employed?"} is supplied by transaction time. Since there may have been changes to the database after 2010, the answers to these questions may be different. The combination of valid time and transaction time supplies information on the form \textit{``In 2015, where did the database believe Taylor was employed during 2010?"}.

To further clarify the concepts introduced above, a detailed example for total permanent disability insurance (or critical illness insurance) is given in Example~\ref{ex:BitemporalDataIbnrRbns}, while Example~\ref{ex:BitemporalDataDisabilityType} is devoted to disability insurance with coverage that depends on the origin of disability.

\begin{example} (Bi-temporal insurance data: Total permanent disability insurance.) \label{ex:BitemporalDataIbnrRbns} \vspace{0.15cm} \\
Consider the following scenario: On 1/1/2020, Taylor buys a total permanent disability insurance effective immediately with a risk period of one year, which pays a sum $b$ if they become disabled before the end of the risk period. For this, Taylor agrees to pay premiums at a rate $\pi$ during the risk period while active. Taylor becomes disabled on 1/3/2020 and reports this to the insurer on 1/5/2020, two months later. On 1/6/2020, one month later, the insurer has finished processing the claim and awards Taylor disability benefits. Furthermore, Taylor is reimbursed for the premiums paid between 1/3/2020 and 1/6/2020.

If the insurer uses a bi-temporal database (valid time and transaction time), the database will at 1/6/2020 contain the following entries:

\begin{figure}[H]
\centering
\includegraphics[scale=1.1]{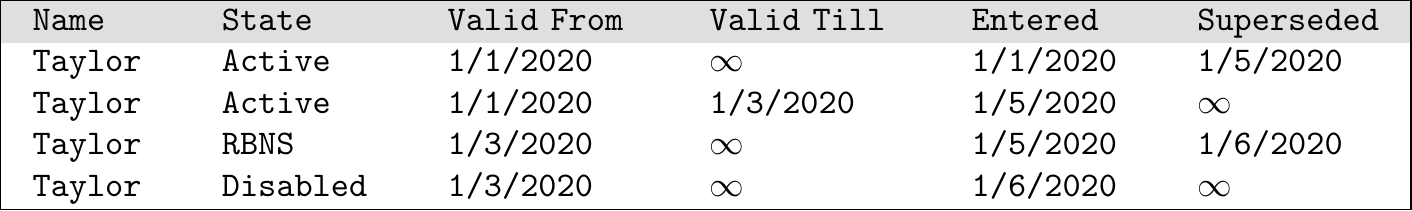}
\end{figure}

\noindent Hence we see that the database records not only what happened in the `real world', but also what was officially recorded at different times. Note that when it is not known when the information is valid till, the database by convention records the timestamp $\infty$. This is likewise the case when it is unknown when the entry will be superseded. Hence, to acquire the most recent belief about when events occurred, one would extract the rows where \texttt{Superseded} was $\infty$.

If the database was uni-temporal (valid time), the entries at 1/6/2020 would be:

\begin{figure}[H]
\centering
\includegraphics[scale=1.1]{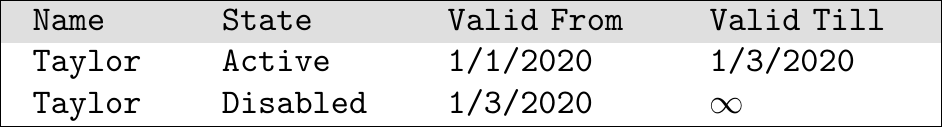}
\end{figure}

\noindent Similarly, at 1/4/2020 the entries would be:

\begin{figure}[H]
\centering
\includegraphics[scale=1.1]{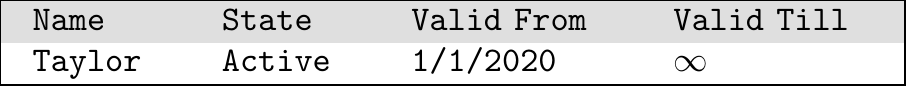}
\end{figure}

\noindent even though Taylor is already disabled at this time, due to the fact that this has not been reported to the insurer yet. 

From these tables, we can see that, as described in~\citet{Snodgrass1985}, different information updates are permitted for a bi-temporal and a uni-temporal database. The uni-temporal database, in contrast to the bi-temporal database, only records what happened in the `real world' based on the newest information. Previous records are modified or deleted. \demoex
\end{example}

\begin{example} (Bi-temporal insurance data: Disability insurance with different origins.) \label{ex:BitemporalDataDisabilityType} \vspace{0.15cm} \\
Consider the following scenario: On 1/1/2020, Jessie buys a disability insurance effective immediately, that pays a rate $b_{i_1}$ if they are affected by a work-related disability (WD) and a rate $b_{i_2}$ if they are affected by a non-work-related disability (NWD). For this, Jessie agrees to pay the premium rate $\pi$ while active. Jessie becomes disabled on 1/5/2020 and reports this to the insurer instantly. The insurer immediately evaluates the disability to have its origin outside of the workplace and therefore pays the rate $b_{i_2}$ starting 1/5/2020. At 1/7/2020, the decision is reevaulated and it is concluded that the disability has its origin at the workplace. Consequently, there is a payment between the insurer and Jessie corresponding to the difference in rates between 1/5/2020 and 1/7/2020, and onward Jessie receives the rate $b_{i_1}$. Nothing else occurs before 1/1/2021.    

If the insurer uses a bi-temporal database (valid time and transaction time), the database will at 1/1/2021 contain the following entries:

\begin{figure}[H]
\centering
\includegraphics[scale=1.1]{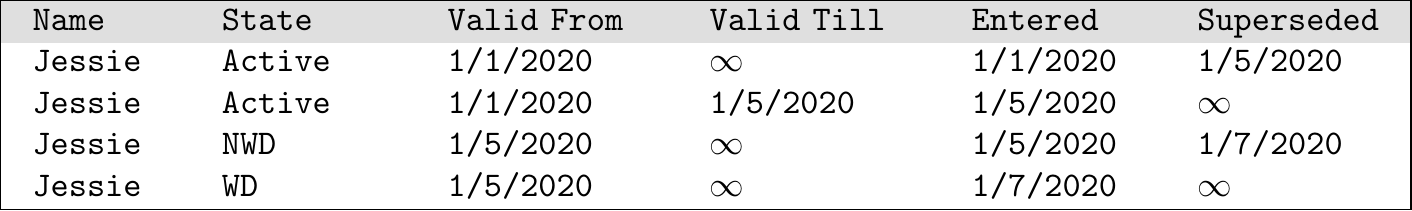}
\end{figure}

\noindent If the database was uni-temporal (valid time), the entries at 1/6/2020 would be:
\begin{figure}[H]
\centering
\includegraphics[scale=1.1]{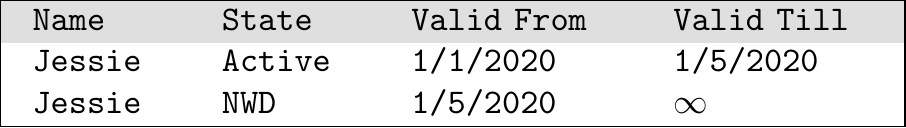}
\end{figure}

\noindent At 1/8/2020 the entries would be:

\begin{figure}[H]
\centering
\includegraphics[scale=1.1]{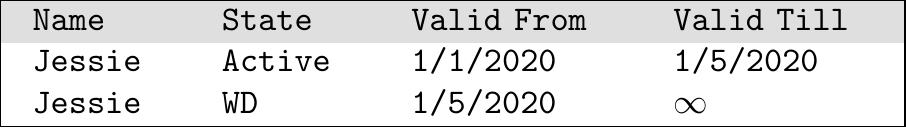}
\end{figure}

\noindent Similar to Example~\ref{ex:BitemporalDataIbnrRbns}, we see a clear need for a bi-temporal database compared to a uni-temporal database. \demoex
\end{example}

\noindent In Example~\ref{ex:BitemporalDataIbnrRbns}, the claim is first an IBNR and later an RBNS claim, while in Example~\ref{ex:BitemporalDataDisabilityType}, the claim is an RBNS claim, but since reporting occurred with no delay, it is not an IBNR claim beforehand. In the following, we illustrate our methods and results on the latter example. Our methods are, however, not constrained to bi-temporal data on the above form, but may be applied to essentially any kind of bi-temporal insurance data.

Practitioners should be well-acquainted with bi-temporal insurance data. Bi-temporal data is important for internal use, as it is needed for reproducibility of statistical analyses when these are based on queries to a database. This is because reproducibility requires rollback information, since one has to recreate the information that the database had at a previous point in time. It also enables one to understand the difference between two otherwise identical analyses, performed at two different points in time. It is also important for external use, since auditors and regulatory authorities may inquire about financial reports from foregone years, making it important for insurers to be able to recreate the prerequisites that a given financial report was based on. As an example of this, Danish life insurance companies are required by law to publish all figures in the balance sheet and income statement of their financial reports for both the current and the previous year. Key figures have to be reported for the past five years. If prior financial reports have been affected by serious errors, the newest report has to publish figures for previous years as if the error had not been committed, so long it is practically feasible, cf.\ \S~86 of~\citet{Erhvervsministeriet2015}.

\subsubsection*{Valid and transaction time information in insurance}

Inspired by the valid and transaction time taxonomy introduced above, and with a slight abuse of the terminology, we define a \textit{valid time process} as a stochastic process that represents the true historical information. We use the notation $X_s$ for the value of the valid time process at time~$s$. With another slight abuse of the terminology, we define a \textit{transaction time process} as a bi-temporal stochastic process that represents both historical and rollback information. We use the notation $X^t_s$ for the value of the process at time~$s$ based on the observations available at time~$t$. When referring to models based on valid time and transaction time processes, we use the terms \textit{valid time model} and \textit{transaction time model}, respectively. Current multi-state models for the biometric state of the insured are seen to be valid time models, as they model a process that describes when events occur without regard to when that information is observed by the insurer. However, in practice there is typically at least some delay in information concerning occurrence, reporting and settlement of claims. This necessitates additional model components, say in the form of models for IBNR and RBNS claims. Such model components are called claims reserving models or simply IBNR and RBNS models. 

Claims reserving based on individual claims data have been subject to much study  in non-life insurance, see e.g.~\citet{Norberg1993,Norberg1999}, \citet{Haastrup1996}, \citet{Antonio2014}, \citet{Badescu2016,Badescu2019}, \citet{Lopez2019}, \citet{Bischofberger2020}, \citet{Delong2021}, \citet{Crevecoeur2022a,Crevecoeur2022b}, and~\citet{Okine2022}. Early research, including~\citet{Norberg1993,Norberg1999}, focuses on joint modeling of all aspects of a claim and the subsequent computation of relevant conditional expectations of future cash flows as high-dimensional integrals with respect to the joint distribution. This may pose significant statistical and numerical challenges, and initially only limited attention was given to practical implementation. Recent research, including~\citet{Crevecoeur2022a,Crevecoeur2022b} and~\citet{Okine2022}, instead focuses on how the expected future cash flows depend on previous observations, typically payment times and payment sizes, and the corresponding factorization of the joint distribution into conditional and marginal parts. In particular, this improves interpretability and readily allows for dynamic reserving where current information is incorporated into the best estimate of future liabilities. This research has hitherto largely consisted of proposing statistical procedures and performing data-driven investigations. Regarding the latter, one attempts to identify the best predictive models by exploring which covariates are advantageous to include (thereby also determining the degree of individualization/collectivization) as well as exploring which statistical models to apply (e.g.\ parametric models such as GLMs, non-parametric models such as empirical distributions, or machine learning methods such as neural networks). 

The primary reason that the life insurance literature has not been similarly occupied with finding suitable statistical models is, as explained in Section~\ref{sec:Introduction}, that the payments stipulated in life insurance contracts are specified in terms of an underlying valid time state process. This implies that the conditional distribution of future jump times given previous observations can be obtained by estimating the distribution of the valid time state process, and a model for the conditional distribution of future payment sizes then follows automatically.  Since the state process may almost always be formulated as a piecewise deterministic process, a point which we discuss in more detail later, the estimation may be performed using standard methods from survival and event history analysis. The state process formulation also allows for more explanatory models compared with the purely predictive models of non-life insurance that result from modeling the transaction time payment process directly.

The problem with the current approach in life insurance of ignoring the transaction time state process is that the alternative state process, namely the valid time state process, is not directly observable. In practice, one partly accounts for this via improvised IBNR corrections and RBNS corrections (on an aggregate level). Compared to the vast literature on IBNR and RBNS models in non-life insurance, the corresponding problem of claims reserving in life insurance has received limited attention. In the following, we seek to amend this.

Since the payments stipulated in a life insurance contract are generated by a valid time process, while the actually observed data is generated by a transaction time process, it becomes essential to establish an in-depth understanding of the relation between valid and transaction time processes in life insurance. The main purpose of this paper is to develop a conceptual and mathematical framework where this relation can be formulated and explored.

\subsubsection*{Piecewise deterministic processes}

In order to achieve as much generality as possible, we take \emph{piecewise deterministic processes} (PDPs) as the starting point for our models. The data generated by a stochastic process and recorded in a database only uniquely determines the path of the stochastic process if the stochastic process is a PDP. Informally, these are processes that have finitely many jumps on finite time intervals and which develop deterministically between the random jump times. This is because the value of the process is only recorded in the database at certain discrete times (e.g.\ daily, monthly, or at jump times), from which the entire path of the process must be inferred. If the data stems from e.g.\ a Brownian motion, which is not piecewise deterministic, then the data can only provide an approximation of the path taken, and the gaps between recorded values must be filled using some algorithm, e.g.\ modeling the process to be linearly evolving between recorded values. In this sense, PDPs provide the most general class of processes for which valid time and transaction time processes can be constructed.

We now define piecewise deterministic processes following Subsection~3.3 in~\citet{Jacobsen2006}.  To keep the exposition from becoming unnecessarily technical, we omit some mathematical details such as the exact form of measurability for certain functions. The interested reader may consult Section 3 of~\citet{Jacobsen2006} for an explicit construction of the basic measurable spaces. In general, further mappings from these spaces are taken to have the image of the mappings equipped with the pushforward $\sigma$-algebra as the codomain space. Let the background probability space be denoted by $(\Omega,\mathbb{F},\mathbb{P})$. Let $(E,\mathcal{E})$ be a measurable space called the \textit{mark space}. Introduce the \textit{irrelevant} mark $\nabla$ denoting the mark of a jump that does not occur in finite time, and set $\overline{E}=E \cup \{ \nabla \}$.

\begin{definition} (Simple point processes.) \vspace{0.15cm} \\
A \textit{simple point process} (SPP) is a sequence $\mathcal{T}=(T_n)_{n \in \mathbb{N}}$ of $[0,\infty]$-valued random variables defined on $(\Omega,\mathbb{F},\mathbb{P})$ such that
\begin{enumerate}[label=(\roman*)]
    \item $\mathbb{P}(0 < T_1 \leq T_2 \leq \dots ) = 1$
    \item $\mathbb{P}(T_n < T_{n+1}, T_n < \infty) = \mathbb{P}(T_n < \infty)$
    \item $\mathbb{P}(\lim_{n \rightarrow \infty} T_n = \infty ) = 1$.\demodef
\end{enumerate}
\end{definition}

If condition (iii) is removed, one obtains the class of simple point processes allowing for \textit{explosion}. In this paper, we limit the study to processes without explosion.

\begin{definition} (Marked point processes.) \vspace{0.15cm} \\
A \textit{marked point process} (MPP) with mark space $E$ is a double-sequence $(\mathcal{T},\mathcal{Y})= \\ ((T_n)_{n \in \mathbb{N}},(Y_n)_{n \in \mathbb{N}})$
of $(0,\infty]$-valued random variables $T_n$ and $\overline{E}$-valued random variable $Y_n$ defined on $(\Omega,\mathbb{F},\mathbb{P})$ such that $\mathcal{T}=(T_n)_{n \in \mathbb{N}}$ is an SPP and such that
\begin{enumerate}[label=(\roman*)]
    \item $\mathbb{P}(Y_n \in E, T_n < \infty) = \mathbb{P}(T_n < \infty)$
    \item $\mathbb{P}(Y_n = \nabla, T_n = \infty) = \mathbb{P}(T_n = \infty)$.\demodef
\end{enumerate}
\end{definition}

\noindent For a given MPP $(\mathcal{T},\mathcal{Y})$, define
\begin{align*}
\langle t \rangle = \sum_{n=1}^\infty 1_{(T_n \leq t)}
\end{align*}
being the number of events in the time interval $[0,t]$, and define
\begin{align*}
H_{t} = (T_1,...,T_{\langle t \rangle}; Y_1,...,Y_{\langle t \rangle})
\end{align*}
being the jump times and marks observed up until and including time~$t$. We refer to $H_t$ as the \textit{MPP history} at time~$t$.

\begin{definition} (Piecewise deterministic process.) \vspace{0.15cm} \\
A piecewise deterministic process with state space $(E,\mathcal{E})$ is an $E$-valued stochastic process $X$ satisfying
\begin{align*}
X_t = f^{\langle t \rangle}_{H_t \mid x_0}(t),
\end{align*}
where $X_0 = x_0$ is non-random, and for every $n \in \mathbb{N}_0$, $f^{n}_{h_n \mid x_0}(t)$ is a measurable $E$-valued function of $h_n=(t_1,...,t_n;y_1,...,y_n)$ with $t_n < \infty$, of $t \geq t_n$, and of $x_0$, satisfying the conditions
\begin{align*}
f^{n}_{h_n \mid x_0}(t_n) = y_n
\end{align*}
for all $h_n$ and $f^{0}_{ \mid x_0}(0) = x_0$.\demodef
\end{definition}

To ensure the existence of relevant conditional distributions going forward, we henceforth assume that the mark space $(E,\mathcal{E})$ is a Borel space. We refer to the functions $f^{n}_{h_n \mid x_0}$ as \textit{evolution functions} of $X$. The explicit dependence on $n$ is standard and serves to highlight that the functional relation may depend on the cardinality of $h_n$. However, to ease notation, we henceforth suppress $x_0$ and $n$, writing for instance simply $f_{h_n}$. It is easily seen that the class of piecewise deterministic processes encompasses the usual choices for the biometric state process, cf.~Example~\ref{ex:PDPStateProcess},~\ref{ex:PDMPStateProcess},~\ref{ex:PDMPMarkovChain}, and~\ref{ex:PDMPSemiMarkov}. 
\begin{example} (Pure jump process.) \label{ex:PDPStateProcess} \vspace{0.15cm} \\
Let $(E,\mathcal{E})$ be a Borel space, and let $X$ be a pure jump process taking values in $E$ which models the biometric state process of the insured. Assume that the initial state of $X$ is fixed, and denote this initial state by $x_0$. A pure jump process is here taken to mean a càdlàg stochastic process with only finitely many jumps in any finite time interval satisfying
\begin{align*}
X_t = \sum_{0 < s \leq t} \Delta X_s
\end{align*}
for $\Delta X_s = X_s-X_{s-}$. We define an MPP $(\mathcal{T},\mathcal{Y})=((T_n)_{n \in \mathbb{N}},(Y_n)_{n \in \mathbb{N}})$ from $X$ by letting $T_n$ be the time of the $n$'th jump of $X$ and setting $Y_n = X_{T_n}$. It is then easy to show that $X$ can be reconstructed from the MPP. To do this, let $f_{h_n}(t) = y_n$, and let $H$ be the MPP history of $(\mathcal{T},\mathcal{Y})$. Then
\begin{align*}
 f_{H_t}(t) = Y_{\langle t \rangle} = X_{T_{\langle t \rangle}} =  X_t,
\end{align*}
so $X$ is a PDP. \demoex
\end{example}

\begin{example} (Pure jump Markov process.) \label{ex:PDMPStateProcess} \vspace{0.15cm} \\
Following the previous example, let $\mathcal{F}^X_t = \sigma(X_s, \: 0 \leq s \leq t)$ be the filtration generated by $X$. For computational tractablility, one often assumes that $X$ is a \textit{Markov process}, meaning
for every $s \leq t$ and $C \in \mathcal{E}$:
\begin{align*}
\mathbb{P}(X_t \in C \mid \mathcal{F}^X_s) = \mathbb{P}(X_t \in C \mid X_s)
\end{align*}
$\mathbb{P}$-a.s., so that the behavior of the process at a future time point is only dependent on the past behavior of the process through the current state of the process. This is called the \textit{Markov property}. Pure jump Markov processes $X$ are a generalization of the usual choices of the state process $Z$ found in the multi-state life insurance literature, as is shown in Example~\ref{ex:PDMPMarkovChain} and Example~\ref{ex:PDMPSemiMarkov}.
\demoex
\end{example}

\begin{example} (Continuous-time Markov chain.) \label{ex:PDMPMarkovChain} \vspace{0.15cm} \\
A continuous-time Markov chain $Z$ with fixed initial state is defined as a piecewise constant Markov process, see e.g.\ Chapter 2 of~\citet{Norris1998} or Section 7.2 of~\citet{Jacobsen2006}. In the life insurance literature, the Markov chain is usually assumed to take values in a finite state space, see e.g.\ \citet{Norberg1991}. A continuous-time Markov chain $Z$ on a finite state space $\{1,2,...,J\}$ for  $J \in \mathbb{N}$ can thus be constructed from a pure jump Markov process $X$ with $E=\{1,2,...,J\}$ by simply setting $Z=X$. \demoex
\end{example}

\begin{example} (Continuous-time semi-Markov process.) \label{ex:PDMPSemiMarkov} \vspace{0.15cm} \\
In many applications, the most recent jump time contains valuable information and hence it may be necessary to include it as a coordinate of $X$ in order to ensure that the Markov property of $X$ holds. For a pure jump process $Z$, define $W$ as the time of the last jump  
\begin{align*}
W_t = \sup \{ 0 \leq s \leq t : Z_s \neq Z_t \}.
\end{align*}
A continuous-time semi-Markov process $Z$ is defined as a pure jump process on a finite state space with the property that $(Z,W)$, equivalently $(Z,U)$ with
\begin{align*}
U_t = t-W_t
\end{align*}
the duration of sojourn in the current state, is a Markov process, see e.g.\ Section~2D in~\citet{Helwich2008}. Semi-Markov models were introduced to life insurance independently by Janssen and Hoem, see~\citet{Janssen1966} and~\citet{Hoem1972}. In the life insurance literature, the parameterization $(Z,U)$ is more common than $(Z,W)$.

For $K \in \mathbb{N}$ define the projection functions $\pi_k : \mathbb{R}^K \mapsto \mathbb{R}$ via $\pi_k (x_1,x_2,...,x_K) = x_k$. If $X$ is a pure jump Markov process with $E=\{1,2,...,J\} \times [0,\infty)$ that satisfies $\pi_2 X_t = T_{\langle t \rangle}$, then a continuous-time semi-Markov process $Z$ may be constructed from $X$ by setting $Z=\pi_1 X$.\demoex
\end{example}

\section{Valid time model} \label{sec:VTM} 
We now introduce the classic multi-state life insurance models. These are valid time models, cf.\ the discussion in Section~\ref{subsubsec:ValidTransactData}. The valid time setup described below is essentially standard in the multi-state life insurance literature, although it is usually only formulated for Markov or semi-Markov processes, see the classics \citet{Hoem1969}, \citet{Hoem1972}, and \citet{Norberg1991}.

\subsubsection*{State process} 
Let the valid time state process $(X_t)_{t \geq 0}$ be an $\mathbb{R}^d$-valued stochastic process for $d \in \mathbb{N}$. We assume that $X$ is a PDP with fixed initial state $x_0$. The jump times are denoted by $\tau_n$.
Let
\begin{align*}
\llangle t \rrangle = \sum_{n=1}^\infty 1_{(\tau_n \leq t)}
\end{align*}
denote the number of jumps by time~$t$ and denote by 
\begin{align*}
H_{t} = (\tau_1,...,\tau_{\llangle  t \rrangle}; X_{\tau_1},...,X_{\tau_{ \llangle t \rrangle } } )
\end{align*}
the MPP history of the process $X$ at time~$t$. The symbols $\tau_n$ and $\llangle t \rrangle$ are used here instead of $T_n$ and $\langle t \rangle$ introduced in the PDP section as the latter are reserved for the transaction time process introduced in Section~\ref{sec:TTM}. Further, let $f_{h_n}$ be the evolution functions of $X$, so that
\begin{align*}
    X_t = f_{H_{t}}(t).
\end{align*}
The information generated by complete observation of the valid time process is given by the filtration $\mathcal{F}^X_t = \sigma( X_s , 0 \leq s \leq t )$.

\subsubsection*{Cash flow}
We now define the valid time cash flow that represents the contractual payments. For later uses, especially to define the transaction time cash flow and to formulate links between present values in valid and transaction time, we need the valid time cash flow to be decoupled from the valid time state process. Denote by $\mathcal{A} \subseteq 2^{[0,\infty)}$ some sufficiently regular collection of sets. For the purposes of this paper, it is sufficient that $\mathcal{A}$ contains the intervals on $[0,\infty)$. Write $x_A$ for $(x_s)_{s \in A}$ with  $A \in \mathcal{A}$. We restrict our attention to $x_A$ for which there exist $(x_s)_{s \in A^c}$ such that $(x_s)_{s \geq 0}$ lies in the image of $X$; here $A^c$ denotes the complement of $A$. Similarly, write $X_A$ for $(X_s)_{s \in A}$. Assume the existence of measurable functions 
\begin{align*}
(x_A,t) \mapsto B(x_A,t) \in \mathbb{R}
\end{align*}
for $t \geq 0$. We interpret $B(x_A,t)$ as the payments generated by the path $x$ on $A \cap [0,t]$. Assume that $t \mapsto B(x_A,t)$ is a càdlàg finite variation function for any $x_A$, so that the measures $B(x_A,\diff t)$ are well-defined. We further assume that the composition $B(X_A,t)$ is incrementally adapted to $\mathcal{F}^X$, meaning that $B(X_A,t)-B(X_A,s)$ is $\sigma(X_v,v \in (s,t])$-measurable for any interval $(s,t] \subseteq [0,\infty)$, cf.\ Definition~2.1 in~\citet{Christiansen2021Time}. For shorthand, we write $B(\diff t) = B(X_{[0,\infty)},\diff t)$. The assumption of incremental adaptedness states that the aggregated payments in $(s,t]$ only depends on the path of $X$ on $(s,t]$. The property incremental adaptedness does not have a critical technical function in this paper, but it clarifies the role of the state process $X$ as the process that determines the contractual payments at a given instance in time. Therefore, we may discuss changes in information and payments through changes to $X$, which allows for practical interpretations that are more intuitive. From a purely mathematical point of view, one could also have based the analysis on the underlying MPP. We name $(B(t))_{t \geq 0}$ the accumulated cash flow in valid time. Note the use of the symbol $B$ for two different objects, namely the stochastic process $t \mapsto B(t)$ and the deterministic function $t \mapsto B(x_A,t)$. It should always be clear from the context which object we are referring to. 

To allow for the valuation of cash flows, we need the time value of money. A detailed treatment of this financial constituent of the model may be found in~\citet{Norberg1990}. Let $t \mapsto \kappa(t)$ be some deterministic strictly positive càdlàg \textit{accumulation function} with initial value $\kappa(0)=1$. The corresponding \textit{discount function} is $t \mapsto \frac{1}{\kappa(t)}$. We let $x_A \mapsto B^\circ( x_A )$ be the time $0$ value of the accumulated payments for the path $x_A$, i.e.\ we set
\begin{align*} 
     B^\circ(x_A) = \int_{(0,\infty)} \frac{1}{\kappa(v)} B(x_A,\diff v),
\end{align*}
presupposing that this object exists (is finite).

\begin{remark} (Cash flow terminology.) \vspace{0.15cm} \\
In the life insurance literature, the stochastic process $(B(t))_{t \geq 0}$ defined above is sometimes also referred to as the payment function, the stream of net payments, the payment process, or the stochastic cash flow, see e.g.~\citet{Norberg1991} and~\citet{BuchardtMollerSchmidt2015}. We use the terminology that $B(t)$ is the accumulated cash flow at time~$t$ while the stochastic measure $B(\diff t)$ is the cash flow. The latter should not be confused with the expected cash flow $A_x(t,\diff s)$ defined by $A_x(t,s) = \mathbb{E}[B(s)-B(t) \mid X(t)=x], \: s \geq t$. In the literature, the expected cash flow is sometimes also ambiguously referred to as the cash flow.
\demormk
\end{remark}

\begin{example} (Cash flow for the usual choices of state processes.) \label{ex:MarkovSemiMarkovCashFlow} \vspace{0.15cm} \\ 
For a set $A$ on the form $[0,v)$, $[0,v]$, or $[0,\infty)$ for $v \geq 0$ and a pure jump Markov process $X$ on a finite state space $E=\{1,2,...,J\}$, see also Example~\ref{ex:PDMPMarkovChain}, one usually specifies the payments as
\begin{align*}
    B(x_A,\diff t) &= \sum_{j = 1 }^J 1_A(t) 1_{(x_{t-} = j )} B_j(\diff t) + \sum_{ \substack{j,k = 1 \\ j \neq k } }^J 1_A(t) b_{jk}(t) n_{jk}(x_A,\diff t),
\end{align*}
where $n_{jk}(x_A,t) = \# \{ s \in [0,t] \cap A : x_{s-}=j, x_s = k \}$, while $t \mapsto B_j(t)$ are càdlàg finite variation functions modeling sojourn payments and $t \mapsto b_{jk}(t)$ are finite-valued Borel-measurable functions modeling transition payments. In this case,
\begin{align*}
    B(\diff t) = \sum_{j = 1 }^J 1_{(X_{t-} = j )} B_j(\diff t) + \sum_{ \substack{j,k = 1 \\ j \neq k } }^J b_{jk}(t) N_{jk}(\diff t)
\end{align*}
for $N_{jk}(t) = n_{jk}(X_{[0,\infty)},t)$. In the semi-Markov case of Example~\ref{ex:PDMPSemiMarkov}, where $X=(Z,W)$ is a pure Markov jump process (with $W$ the time of the last jump), one usually specifies the payments as 
\begin{align*}
    B(x_A,\diff t) &= \sum_{j = 1 }^J 1_A(t) 1_{(z_{t-} = j )} B_{j,w_{t-}}(\diff t) + \sum_{ \substack{j,k = 1 \\ j \neq k } }^J 1_A(t) b_{(j,w_{t-})(k,t)}(t) n_{jk}(z_A,\diff t),
\end{align*}
where $x_s = (z_s,w_s)$, while $t \mapsto B_{j,w}(t)$ and $t \mapsto b_{(j,w)(k,t)}(t)$ satisfy the same regularity conditions as in the Markov case. Then
\begin{align*}
    B(\diff t) = \sum_{j = 1 }^J 1_{(Z_{t-} = j )} B_{j,W_{t-}}(\diff t) + \sum_{ \substack{j,k = 1 \\ j \neq k } }^J b_{(j,W_{t-})(k,t)}(t) N_{jk}(\diff t)
\end{align*}
for $N_{jk}(t) = n_{jk}(Z_{[0,\infty)},t)$.\demoex
\end{example}

\begin{example} (Continuous compound interest.) \label{ex:BankAccount} \vspace{0.15cm} \\
Under continuous compound interest with force of interest $t \mapsto r(t)$, a deposit of one unit currency in a savings account at time~$0$ has at time~$t$ accumulated to
\begin{align*}
\kappa(t) = \exp \bigg( \int_{(0,t]} r(s) \diff s \bigg),
\end{align*}
see e.g.\ \citet{Norberg1990}.\demoex
\end{example}

\begin{example} (Disability insurance with different origins: Valid time model.) \label{ex:SimpleValidTime} \vspace{0.15cm} \\
We construct a valid time model for the product described in Example~\ref{ex:BitemporalDataDisabilityType}. The state of the insured is modeled as a semi-Markov process $X=(Y,U)$, cf.\ Example~\ref{ex:PDMPSemiMarkov}, where $Y$ takes values in the state space depicted in Figure~\ref{fig:simpleValidModel}.
\begin{figure}[H]
	\centering
	\scalebox{0.75}{
	\begin{tikzpicture}[node distance=4em, auto]
	\node[punkt] (g) {active};;
	\node[punktl, right=of g] (i1) {nwd};
    \node[punktl, below=2em of i1] (i2) { wd };
    \node[punktl, right=of i1] (r) {reactivated};
    \node[right=-0.75em of i1] (dummyi1) {};
    \node[right=-0.2em of i2] (dummyi2) {};
	\node[punkt, below =2em of i2] (dead) {dead};
     \node[anchor=north east, at=(g.north east)] {$a$};
     \node[anchor=north east, at=(i1.north east)] {$i_1$};
     \node[anchor=north east, at=(i2.north east)] {$i_2$};
     \node[anchor=north east, at=(r.north east)] {$r$};
     \node[anchor=north east, at=(dead.north east)] {$d$};
	\path 
        ($(g.east)$) edge [pil, bend right = 15] ($(i1.west)$)
        ($(g.east)$) edge [pil, bend right = 15] ($(i2.west)$)
        ($(dummyi1.east)$) edge [pil, bend right = 15] ($(r.west)$)
        ($(dummyi2.east)$) edge [pil, bend right = 15] ($(r.west)$)
        ($(g.south east)$) edge [pil, shadecolor, bend right = 30] ($(dead.west)$)
        ($(r.south west)$) edge [pil, shadecolor, bend left = 30] ($(dead.east)$)
        ($(i2.south east)$) edge [pil, shadecolor, bend left = 30] ($(dead.north east)$)
        ($(i1.south west)$) edge [pil, shadecolor, bend right = 45] ($(dead.north west)$)
	; %+(0.1,0)
    \end{tikzpicture}
    }
	\caption{The $Y$-component of the biometric state process $X$ takes values in $\{a,i_1,i_2,r,d\}$. The absence of an arrow between states indicates that a direct jump between these states is impossible. To reduce clutter, the arrows into the dead state are made semi-transparent.}
	\label{fig:simpleValidModel}
\end{figure}
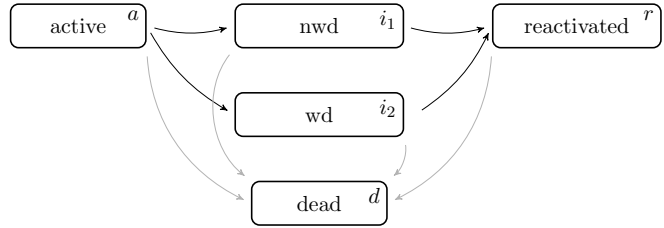
\noindent The cash flow is assumed to consist of a premium rate $\pi < 0$ while active and a disability rate $b_{i_k} > 0$ while disabled in state $i_k$:
\begin{align*}
    B(x_A,\diff t) &= 1_A(t) \pi 1_{(y_{t-}=a)} \diff t  + 1_A(t) \sum_{k=1}^2  b_{i_k} 1_{(y_{t-}=i_k)} \diff t,
\end{align*}
so that
\begin{align*}
    B(\diff t) = \pi 1_{(Y_{t-}=a)} \diff t  +  \sum_{k=1}^2  b_{i_k} 1_{(Y_{t-}=i_k)} \diff t.
\end{align*}
The deterministic valid time cash flow $B(x_A,\diff t)$ is used to specify the transaction time cash flow in Example~\ref{ex:SimpleTransactionTime}. Since there are no payments in the reactivated or dead states, the semi-Markov assumption is, for valuation purposes, actually not a restriction. \demoex
\end{example}

\section{Transaction time model} \label{sec:TTM}
We now introduce the transaction time models corresponding to the valid time models introduced in Section~\ref{sec:VTM}. The transaction time setup described below is novel, and it allows for cash flows that are tailored to describe the payments that occur in real-time, like the ones that would result from the cases described in Examples~\ref{ex:BitemporalDataIbnrRbns} and~\ref{ex:BitemporalDataDisabilityType}.

\subsubsection*{State process}

Let $(Z_t)_{t\geq0}$ be an $\mathbb{R}^q$-valued pure jump process for $q \in \mathbb{N}$. We think of $Z$ as describing the claim settlement process of the policy, which is observed in real-time by the insurer. The process $Z$ jumps to another state when new information is made available to the insurer, for example through communication with the insured or due to internal decisions from the insurer. The jump times are denoted by $T_n$. We write
\begin{align*}
\langle t \rangle = \sum_{n=1}^\infty 1_{(T_n \leq t)}
\end{align*}
for the number of jumps of $Z$ that have happened at time $t$. We further introduce a doubly indexed stochastic process $H^t_s$ for $0 \leq s \leq t$, which we name the \emph{transaction time MPP history corresponding to $H_s$}. The idea is to interpret $H^t_s$ as the value of $H_s$ based on the transaction time information available at time $t$. We extend the definition to $s > t$ by letting $H^t_s = H^t_t$, akin to how in Example~\ref{ex:BitemporalDataIbnrRbns}, the most recent history is set to be valid until $\infty$. The role of $Z$ is to contain transaction time information that may not yet have resulted in a change to the transaction time MPP. An example could be a reported claim occurrence that has yet to be awarded or rejected. We briefly note that the main example contained of this paper is sufficiently simple that one could have omitted $Z$ in the formulation and simply used $H^t_s$. This is, however, not possible in the general case. To imbue $H^t_s$ and $Z$ with the desired interpretations outlined above, we specify the dependencies to the corresponding valid time process:
\begin{enumerate}[label=(\roman*)]
    \item We only allow changes to the transaction time MPP history corresponding to $H_s$ to occur at jumps of $Z$, so the process $Z$ is the driver of new transaction time information arriving. This is formalized as 
\begin{align*}
H_s^t = H_s^{T_{\langle t \rangle }}
\end{align*}
for all $t,s \geq 0$.
    \item We assume that there is a finite time after which no new information arrives (e.g.\ the time of death of the insured). This is formalized as
\begin{align*}
Z_t = Z_\eta
\end{align*}
    for all $t \geq \eta$ with $\mathbb{P}(\eta < \infty) = 1$. This condition is satisfied in any practical application. We name $\eta$ the \textit{absorption time}.
    \item When there are no future changes to the transaction time MPP history corresponding to $H_s$, the observations are taken to be the true historical information. These observations constitute the finalized timeline that the insurer will observe. This is formalized as
\begin{align*}    
H_s^t = H_s
\end{align*}
for all $t \geq \eta$ and $s \geq 0$. This is the fundamental link between the valid time and transaction time models.
\end{enumerate}
We name these assumptions the \textit{basic bi-temporal structure assumptions}.

The transaction time state process is then defined as the doubly indexed stochastic process $X^t_s$ satisfying
\begin{align*}
X_s^t = f_{H^t_s}(s),
\end{align*}
with the interpretation that $X_s^t$ is the value of $X_s$ based on the available transaction time information at time~$t$.

At time~$t$, the insurer has observed $H^{s}_{s}$ and $Z_{s}$ for all $s \leq t$. The insurer's available information is therefore generated by a process $(\mathcal{Z}_t)_{t\geq0}$ given by
\begin{align*}
t \mapsto \mathcal{Z}_t = \big( Z_{t}, H^{t}_{t}  \big),
\end{align*}
and we define the \textit{transaction time information} to be the filtration $\mathcal{F}^\mathcal{Z}_t = \sigma( \mathcal{Z}_s , 0 \leq s \leq t )$. Note that $\mathcal{Z}$ is a PDP by Example~\ref{ex:PDPStateProcess}, since it is a pure jump process which furthermore can be embedded into the Borel space $(\mathbb{R}^\infty, \mathbb{B}(\mathbb{R}^\infty))$.  

\begin{remark} (Transaction and valid time filtrations.) \label{remark:TTandVTFiltration} \vspace{0.15cm} \\
Note that in general, it holds that $\mathcal{F}^{\mathcal{Z}}_t \not \subseteq \mathcal{F}^X_t$ and $\mathcal{F}^X_t \not \subseteq \mathcal{F}^{\mathcal{Z}}_t$. This corresponds to $(\mathcal{Z}_s)_{0\leq s \leq t}$ not being known from $(X_s)_{0\leq s \leq t}$ and vice versa. In other words, the same valid time realizations can stem from different transaction time realizations, and transaction time realizations in a period do not generally determine the valid time realizations in that same period, due to the possibility of new information arriving later.\demormk
\end{remark}

\subsubsection*{Cash flow}

We now define the transaction time cash flow that represents the payments occurring in real-time. We denote the accumulated cash flow in transaction time by $(\mathcal{B}(t))_{t \geq 0}$ and define it as
\begin{equation*}
    \begin{aligned}[c]
    &\mathcal{B}(\diff t) = B(X^{t}_{[0,\infty)}, \diff t) + \diff \bigg( \sum_{0 < s \leq t} \kappa(s)  \Big( B^\circ(X^s_{[0,s)} )-B^\circ( X^{s-}_{[0,s)}) \Big) \bigg), \hspace{0.5cm} \mathcal{B}(0)=B(0).
    \end{aligned}
\end{equation*}
The first term is well-defined as it may be written as $\sum_{n=0}^\infty 1_{(\langle t \rangle = n)} B(X^{T_n}_{[0,\infty)}, \diff t)$ with the convention $T_0=0$. Hence the cash flow in transaction time consists of running payments similar to the ones in the valid time model, but here determined by $X^t$ instead of $X$, as well as lump sum payments when some past $X$ values are changed based on the newest transaction time information. The lump sum payments are commonly known as \textit{backpay}. Backpay makes the accumulated historic payments congruent with the latest MPP history. Furthermore, and closely connected to the principle of no arbitrage, the backpay is accumulated to the time of payout according to $\kappa$, so that the insured is no better or worse off than if the payment had not been delayed and had been deposited in a savings account immediately after payout. Note also that the transaction time cash flow is incrementally adapted to $\mathcal{F}^\mathcal{Z}$, meaning that $\mathcal{Z}$ is the process that determines the transaction time payments at a given instance in time.

Note that since the first term of $\mathcal{B}$ is evaluated in $X^t$ and not $X^{t-}$, the payment at time~$t$ based on the most recent information is included in the first term and should not be included in the backpay, which is why the right endpoint is excluded in the interval $[0,s)$ that appears in the second term. Note also that $X^s_{[0,s)}=X^{s-}_{[0,s)}$ unless $Z$ jumps at time $s$. This implies that backpay can only be paid at jumps of $Z$.

\begin{remark} (Cash flow modeling in non-life insurance.) \label{remark:CashFlowNonLife} \vspace{0.15cm} \\
As described in Section~\ref{sec:VTTT}, using non-life insurance methods, one would model the expected future payments arising from the real-time cash flow $\mathcal{B}$ by disregarding the state process $X^t_s$ and the structure it imposes on the cash flow, and instead find suitable statistical models that predict $\mathcal{B}(\diff t)$ directly.\demormk
\end{remark}

\noindent In Example~\ref{ex:SimpleTransactionTime}, we consider the situation where it could be difficult to determine the origin or cause of a disability. This results in retroactive changes, which we can describe using a transaction time model.

\begin{example} (Disability insurance with different origins: Transaction time model.) \label{ex:SimpleTransactionTime} \vspace{0.15cm} \\
We assume a simple transaction time model, namely that the $Z$-process takes values in the state space illustrated in Figure~\ref{fig:simpleTransactionModel}.
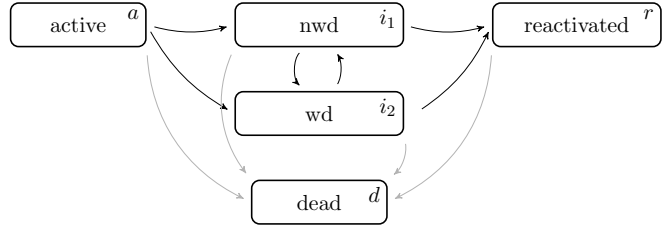
\begin{figure}[H]
	\centering
	\scalebox{0.75}{
	\begin{tikzpicture}[node distance=4em, auto]
	\node[punkt] (g) {active};;
	\node[punktl, right=of g] (i1) {nwd};
    \node[punktl, below=2em of i1] (i2) { wd };
    \node[punktl, right=of i1] (r) {reactivated};
    \node[right=-0.75em of i1] (dummyi1) {};
    \node[right=-0.2em of i2] (dummyi2) {};
	\node[punkt, below =2em of i2] (dead) {dead};
     \node[anchor=north east, at=(g.north east)] {$a$};
     \node[anchor=north east, at=(i1.north east)] {$i_1$};
     \node[anchor=north east, at=(i2.north east)] {$i_2$};
     \node[anchor=north east, at=(r.north east)] {$r$};
     \node[anchor=north east, at=(dead.north east)] {$d$};
	\path 
        ($(g.east)$) edge [pil, bend right = 15] ($(i1.west)$)
        ($(g.east)$) edge [pil, bend right = 15] ($(i2.west)$)
        ($(dummyi1.east)$) edge [pil, bend right = 15] ($(r.west)$)
        ($(dummyi2.east)$) edge [pil, bend right = 15] ($(r.west)$)
        ($(g.south east)$) edge [pil, shadecolor, bend right = 30] ($(dead.west)$)
        ($(r.south west)$) edge [pil, shadecolor, bend left = 30] ($(dead.east)$)
        ($(i2.south east)$) edge [pil, shadecolor, bend left = 30] ($(dead.north east)$)
        ($(i1.south west)$) edge [pil, shadecolor, bend right = 30] ($(dead.north west)$)
        ($(i2.north)+(0.25,0)$) edge [pil, bend right = 30] ($(i1.south)+(0.25,0)$)  
        ($(i1.south)+(-0.25,0)$) edge [pil, bend right = 45] ($(i2.north)+(-0.25,0)$)            
	; %+(0.1,0)
    \end{tikzpicture}
    }
	\caption{The $Z$-component of the transaction time process $\mathcal{Z}$ takes values in $\{a,i_1,i_2,r,d\}$. To reduce clutter, the arrows into the dead state are made semi-transparent.}
	\label{fig:simpleTransactionModel}
\end{figure}

\noindent We further assume that the transitions in the transaction time model equal those in the valid time model, but that the origin of disability $i_k$ is not necessarily correctly identified when the claim is reported. The presumed origin may change while the insured is still disabled, but not after the insured has reactivated or died.

Let $\theta_j =  \inf \{ s \geq 0 : Z_s = j \}$ be first hitting times, let $\theta_i = \min\{\theta_{i_1},\theta_{i_2}\}$ be the first hitting time of a disabled state, and write $i(t)$ for the most recently visited disabled state before time $t$. Define $x \wedge y = \min \{x,y\}$. By the specification of the model, $H^t_s$ contains $(\theta_d,d)$ on $(\theta_d \leq t \wedge s)$, it contains $(\theta_r,r)$ on $(\theta_r \leq t \wedge s)$, and it contains $(\theta_i,i_k)$ on $(\theta_i \leq t \wedge s, i(t) = i_k)$. Finally, let $N^Z_{jk}(t)$ be the number of transitions from state $j$ to state $k$ for the process $Z$ on $[0,t]$. The transaction time cash flow then reads
\begin{align*}
    \mathcal{B}(\diff t) &= \pi 1_{(Y^t_{t-}=a)} \diff t  +  \sum_{k=1}^2  b_{i_k} 1_{(Y^t_{t-}=i_k)} \diff t + \sum_{\substack{j,k \in \{1,2\} \\ j \neq k}} \bigg(\int_{[\theta_i,t)} \frac{\kappa(t)}{\kappa(s)}(b_{i_k} - b_{i_j} ) \diff s  \bigg) N^Z_{i_j i_k}(\diff t). \\[-7mm]\tag*{\demoex}
\end{align*}
\end{example}

\FloatBarrier

\section{Reserving} \label{sec:Reserving}

Having defined the state process and cash flow in the valid and transaction time models, we are now in a position to define the prospective present value and prospective reserve in the respective models. The main result of the paper, Theorem~\ref{theorem:PZequalsPYR}, linking the present values in the two models, is deduced and discussed. Finally, the dynamics of the transaction time reserve is derived and its role in model validation is briefly considered. 

\subsubsection*{Valid time reserve}
Let the classic valid time prospective present value $(P(t))_{t \geq 0}$ be defined by
\begin{align*}
P(t)= \int_{(t,\infty)} \frac{\kappa(t)}{\kappa(s)} B(\diff s).
\end{align*}
We assume that $P(t)$ is well-defined and belongs to $L^1(\Omega,\mathbb{F},\mathbb{P})$ for every $t \geq 0$. Then we can define the corresponding expected present value by
\begin{align}
    V(t) = \mathbb{E}\big[P(t) \mid \mathcal{F}^X_t \big]
\end{align}
for any $t\geq0$. As a function of $t$, this is also known as the \textit{prospective reserve} in valid time. We consider a version of $V$ presumed to be $\mathcal{F}^X$-adapted and $\mathbb{P}$-a.s.\ right-continuous.
\begin{remark} (Reserve for Markov state process.) \label{remark:ReserveMarkovSemiMarkov} \vspace{0.15cm} \\
Note that $P(t)$ is $\sigma(X_s , t \leq s < \infty)$-measurable, so if $X$ is a Markov-process, it follows that
\begin{align*}
V(t) = \mathbb{E}[P(t) \mid X_t ]. \tag*{\demormk}
\end{align*}
\end{remark}

\noindent For random variables $Y_1$ and $Y_2$, we say that $Y_1=Y_2$ on $A \in \mathbb{F}$ if $Y_1(\omega) = Y_2(\omega)$ for $\mathbb{P}$-almost all $\omega \in A$. We immediately have the following representation of the present value:
\begin{proposition} (Valid time present value.) \label{proposition:ValidTimePV} \vspace{0.15cm} \\
For $t \geq 0$, it holds that
\begin{align*}
P(t) = \kappa(t)\Big( B^\circ(X^\eta_{[0,\infty) })- B^\circ(X^\eta_{[0,t]}) \Big)
\end{align*}
on the event $(\eta < \infty)$.
\end{proposition}
\begin{proof}
Using the definitions introduced above, we find that
\begin{align*}
    P(t) &= \kappa(t) \int_{(t,\infty)} \frac{1}{\kappa(s)} B(\diff s) \\[4pt]
    &= \kappa(t) \bigg( \int_{[0,\infty)} \frac{1}{\kappa(s)} B(\diff s) - \int_{[0,t]} \frac{1}{\kappa(s)} B(\diff s) \bigg) \\[4pt]
    &= \kappa(t) \bigg( \int_{[0,\infty)} \frac{1}{\kappa(s)} B(X_{[0,\infty)},\diff s) - \int_{[0,t]} \frac{1}{\kappa(s)} B(X_{[0,\infty)},\diff s) \bigg) \\[4pt]
    &= \kappa(t)\Big( B^\circ(X_{[0,\infty) })- B^\circ(X_{[0,t]}) \Big).
\end{align*}
Now on the event $(\eta < \infty)$ we have that $X_s = X^\eta_s$ for all $s \geq 0$, from which the result follows.
\end{proof}

\subsubsection*{Transaction time reserve}

The transaction time prospective present value $(\mathcal{P}(t))_{t \geq 0}$ is defined as
\begin{align*}
    \mathcal{P}(t) &= \int_{(t,\infty)} \frac{\kappa(t)}{\kappa(s)}  \mathcal{B}(\diff s).
\end{align*}
We assume $\mathcal{P}(t)$ is well-defined and belongs to $L^1(\Omega,\mathbb{F},\mathbb{P})$ for any $t \geq 0$. Let $\mathcal{V}$ be the corresponding expected present value in the transaction time model
\begin{align} \label{eq:VZ}
    \mathcal{V}(t)= \mathbb{E} \big[ \mathcal{P}(t) \mid \mathcal{F}^{\mathcal{Z}}_t \big].
\end{align}
Just as for $V$, we consider a version of $\mathcal{V}$ presumed to be $\mathcal{F}^{\mathcal{Z}}$-adapted and $\mathbb{P}$-a.s.\ right-continuous. Since the backpay $B^\circ(X^s_{[0,s)} )-B^\circ( X^{s-}_{[0,s)})$ telescopes, we also get the following representation of the transaction time present value $\mathcal{P}$: 
\begin{proposition} (Transaction time present value.) \label{proposition:TransactionTimePV} \vspace{0.15cm} \\
For $t \geq 0$, it holds that
\begin{align*}
\mathcal{P}(t)=\kappa(t)\Big( B^\circ(X^\eta_{ [0,\infty)})- B^\circ(X^t_{[0,t]}) \Big)
\end{align*}
on the event $(\eta < \infty)$.
\end{proposition}

\begin{proof}
Note that the equality is trivially satisfied on $(t \geq \eta)$ by Proposition~\ref{proposition:ValidTimePV} and the observation that $P(t)=\mathcal{P}(t)$ on $(t \geq \eta)$ since $X^t = X^\eta$ implies that the running payments agree and there is no backpay after time $t$. Hence, what remains to be shown is that the equality also holds on $(t < \eta)$. We therefore let all the remaining calculations be on $(t < \eta)$. We re-index $T_k=T_{k+\langle t \rangle}$ for $k \in \mathbb{N}$, so $T_k$ now refers to the $k$'th jump of $Z$ after time $t$. Let ${n_\eta}$ be a random variable $\mathbb{P}$-a.s.\ taking values in $\mathbb{N}$, with the defining feature that $ T_{n_\eta} = \eta$, so that ${n_\eta}$ is the number of jumps after time $t$ until $Z$ is absorbed. For notational convenience, introduce
\begin{align*}
    \mathcal{P}^{\circ}(t) &= \frac{1}{\kappa(t)} \mathcal{P}(t)
\end{align*}
and $\beta_s = B^\circ(X^s_{[0,s)} )-B^\circ( X^{s-}_{[0,s)})$.
We then have on the event $(\eta < \infty)$:
\begin{align*}
       \mathcal{P}^{\circ}(t) 
       &= \int_{(t,\infty)} \frac{1}{\kappa(s)} \sum_{n=0}^\infty 1_{(\langle s \rangle = n)} B(X^{T_n}_{[0,\infty)}, \diff s) + \sum_{t < s < \infty} \beta_s \\[4pt]
       &=  \int_{(t,T_1)}  \frac{1}{\kappa(s)} B(X^t_{[0,\infty)}, \diff s) +  \beta_{T_1} \\
       & \quad + \sum_{n=1}^{n_\eta-1} \bigg( \int_{[T_n,T_{n+1})}  \frac{1}{\kappa(s)} B(X^{T_n}_{[0,\infty)}, \diff s) +  \beta_{T_{n+1}} \bigg) \\ 
       &  \quad + \int_{[\eta,\infty)} \frac{1}{\kappa(s)} B(X^\eta_{[0,\infty)}, \diff s)
\end{align*}   
by decomposing the integrals between jumps of $Z$ and using that $\beta$ is only non-zero at jumps of $Z$. Hence we can write
\begin{align*}       
       \mathcal{P}^{\circ}(t) &=  B^\circ(X^{t}_{[0,T_1)})-B^\circ(X^{t}_{[0,t]})  +  B^\circ(X^{T_1}_{[0,T_1)}) - B^\circ(X^{t}_{[0,T_1)} )  \\
       & \quad + \sum_{n=1}^{n_\eta-1}   \Big( B^\circ(X^{T_{n}}_{[0,T_{n+1})})-B^\circ(X^{T_{n}}_{[0,T_n)})  +   B^\circ(X^{T_{n+1}}_{[0,T_{n+1})}) - B^\circ(X^{T_{n}}_{[0,T_{n+1})} ) \Big)  \\
       & \quad +   B^\circ(X^{\eta}_{[0,\infty)}) - B^\circ(X^{\eta}_{[0,\eta)} ) \\[4pt]
       &= B^\circ(X^{T_1}_{[0,T_1)})-B^\circ(X^{t}_{[0,t]}) + \sum_{n=1}^{n_\eta-1}  \Big( B^\circ(X^{T_{n+1}}_{[0,T_{n+1})}) -B^\circ(X^{T_{n}}_{[0,T_n)})  \Big)  \\
       & \quad +   B^\circ(X^{\eta}_{[0,\infty)}) - B^\circ(X^{\eta}_{[0,\eta)} ).
\end{align*}   
Observe that the sum telescopes, so we have
\begin{align*}       
       \mathcal{P}^{\circ}(t) &=  B^\circ(X^{T_1}_{[0,T_1)})-B^\circ(X^{t}_{[0,t]}) + B^\circ(X^{\eta}_{[0,\eta)})-B^\circ(X^{T_1}_{[0,T_1)}) + B^\circ(X^{\eta}_{[0,\infty)}) - B^\circ(X^{\eta}_{[0,\eta)} )  \\
       &=   B^\circ(X^{\eta}_{[0,\infty)}) - B^\circ(X^t_{[0,t]}).
\end{align*}
Consequently, 
\begin{align*}
    \mathcal{P}(t) = \kappa(t) \mathcal{P}^{\circ}(t) = \kappa(t) \Big(  B^\circ(X^{\eta}_{[0,\infty)}) - B^\circ(X^t_{[0,t]})  \Big)
\end{align*}
as desired.
\end{proof}

\subsubsection*{Relation between reserves}

Using Proposition~\ref{proposition:ValidTimePV} and Proposition~\ref{proposition:TransactionTimePV}, the following theorem is now immediate:
\begin{theorem} (Representations of transaction time present value.) \label{theorem:PZequalsPYR} \vspace{0.15cm} \\
For $t \geq 0$, it holds that
\begin{align*}
     \mathcal{P}(t) &=  P(t)+\kappa(t) \Big( B^\circ(X^\eta_{[0,t]}) - B^\circ(X^t_{[0,t]}) \Big) \\[4pt]
     &=   P(t)+ \sum_{t < s < \infty} \kappa(t) \Big( B^\circ(X^s_{[0,t]} )-B^\circ( X^{s-}_{[0,t]}) \Big)  
\end{align*}
on the event $(\eta < \infty)$.
\end{theorem}

\begin{proof}
From Propositions~\ref{proposition:ValidTimePV} and~\ref{proposition:TransactionTimePV}, it follows that on $(\eta < \infty)$:
\begin{align*}
    \mathcal{P}(t)-P(t) &= \kappa(t) \Big( B^\circ(X^\eta_{ [0,t]})- B^\circ(X^t_{[0,t]})  \Big),
\end{align*}
which implies
\begin{align*}
    \mathcal{P}(t) &= P(t) + \kappa(t) \Big(  B^\circ(X^\eta_{ [0,t]})- B^\circ(X^t_{[0,t]})   \Big).
\end{align*}
This proves the first equality. The second equality corresponds to showing that
\begin{align*}
    B^\circ(X^\eta_{[0,t]})-B^\circ(X^t_{[0,t]})= \sum_{t < s < \infty} \Big( B^\circ(X^s_{[0,t]} )-B^\circ( X^{s-}_{[0,t]}) \Big). 
\end{align*}
This is trivially satisfied on $(t \geq \eta)$ since both the left- and right-hand side are zero. Using the same notation as the proof of Proposition~\ref{proposition:TransactionTimePV} and the convention $T_0=t$, we see on $(t < \eta)$:
\begin{align*}
   \sum_{t < s < \infty} \Big( B^\circ(X^s_{[0,t]} )-B^\circ( X^{s-}_{[0,t]}) \Big)
    &= \sum_{n=1}^{n_\eta} \Big(B^\circ(X^{T_n}_{[0,t]})-B^\circ(X^{T_n-}_{[0,t]}) \Big)\\[4pt]
    &= \sum_{n=1}^{n_\eta} \Big(B^\circ(X^{T_n}_{[0,t] })-B^\circ(X^{T_{n-1}}_{[0,t]}) \Big)\\[4pt]
    &= B^\circ(X^\eta_{[0,t]})-B^\circ(X^t_{[0,t]})
\end{align*}
since the sum telescopes. This establishes the second equality and thus completes the proof.
\end{proof}

\noindent The assertion of Theorem~\ref{theorem:PZequalsPYR} is quite intuitive, and one could alternatively have formulated the setup by taking $\mathcal{P}(t) =  P(t)+\kappa(t) ( B^\circ(X^\eta_{[0,t]}) - B^\circ(X^t_{[0,t]}))$ as a definition and proceeded from there. Defining the transaction time present value through the cash flow as $\mathcal{P}(t) = \int_{(t,\infty)} \frac{\kappa(t)}{\kappa(s)}  \mathcal{B}(\diff s)$ seems, however, the more principled approach. Furthermore, the cash flow representation is easier to work with in some situations; confer also with the proof of Theorem~\ref{thm:Dynamics}. Regardless, the representation in Theorem~\ref{theorem:PZequalsPYR} tends to be more convenient when linking the transaction and valid time reserves as can be seen, for instance, in Example~\ref{ex:SimpleReserving}.

\begin{remark} (Stylized illustration of Theorem~\ref{theorem:PZequalsPYR}.) \vspace{0.15cm} \\
To better explain the contents of Theorem~\ref{theorem:PZequalsPYR}, we give a stylized example. Suppose that one is situated at time $t$, that benefits have been paid in $[0,t-1]$, and that the finalized payments consist of benefits during all of $[0, t+1]$. Suppose further that the benefits concerning $(t-1,t+1]$ occur as backpay at time $t+1$. Then the payments concerning $(t,t+1]$ appear in $P(t)$, while the payments concerning $(t-1,t]$ appear in $\kappa(t) \big( B^\circ(X^\eta_{[0,t]}) - B^\circ(X^t_{[0,t]}) \big) =  \kappa(t) \big( B^\circ(X^{t+1}_{[0,t]} )-B^\circ( X^{t+1-}_{[0,t]}) \big)$. \demormk
\end{remark}

\begin{remark} (Relation between reserves in valid and transaction time.) \label{remark:PresentValueValidAndTransactionTime} \vspace{0.15cm} \\
By taking the conditional expectation given $\mathcal{F}^\mathcal{Z}_t$ of the expression from Theorem~\ref{theorem:PZequalsPYR}, we conclude that the expected present value in transaction time $\mathcal{V}$ is different from the classic expected present value in valid time $V$ in two fundamental ways:
\begin{enumerate}
    \item It reserves additionally to previously wrongly settled payments, so it is no longer strictly prospective in valid time, in the sense that payments may relate to valid time events that lie before the current point in time.
    \item It conditions on the filtration $\mathcal{F}^{\mathcal{Z}}$, which is observable, compared to $\mathcal{F}^X$, which is only partially observable.
\end{enumerate}
 If there had been no previously wrongly settled payments, and if the conditional expectations given the two filtrations had been equal, then the reserves would also have been identical. Even though the relation between the present values is relatively simple, this does not translate into a simple relation between $\mathcal{V}$ and $V$ in the general case. This is because we have so far imposed very little structure on the model for $\mathcal{Z}$, so how the conditional distribution of $\mathcal{Z}_s$ given $\mathcal{F}^\mathcal{Z}_t$ for $s \geq t$ depends on $\mathcal{F}^\mathcal{Z}_t$ can be almost arbitrarily complicated.

Note that another important consequence of Theorem~\ref{theorem:PZequalsPYR} (or rather Proposition~\ref{proposition:TransactionTimePV}) is that one does not need the distribution of $\mathcal{Z}$ to calculate $\mathcal{V}(t)$. The conditional distribution of $X_{[0,\infty)}$ given $\mathcal{F}^\mathcal{Z}_t$ is sufficient, since
\begin{align*}
     \mathcal{V}(t) =  \mathbb{E}\big[ \kappa(t) B^\circ(X_{ [0,\infty)}) \mid \mathcal{F}^\mathcal{Z}_t  \big] - \kappa(t) B^\circ(X^t_{[0,t]}).
\end{align*}
Further, since $X_{[0,\infty)}$ equals $X^\eta_{[0,\infty)}$ almost surely with respect to $\mathbb{P}$ and $X^\eta_{[0,\infty)}$ is $\mathcal{F}^\mathcal{Z}_\infty$-measurable, the  distribution of $\mathcal{Z}$ determines the conditional distribution of $X_{[0,\infty)}$ given $\mathcal{F}^\mathcal{Z}_t$ for any outcome of $\mathcal{Z}_{[0,t]}$. Consequently, the conditional distribution might be the natural modeling object. \demormk
\end{remark}

\begin{remark} (Non-monotone information.) \vspace{0.15cm} \\
Write $V(t) = g(X_{[0,t]})$ for a measurable function $g$, which exists by the Doob-Dynkin lemma. Continuing the discussion from Remark~\ref{remark:PresentValueValidAndTransactionTime}, standard practice seems to be to use the individual reserve $V^t(t)=g(X^t_{[0,t]})$ at time $t$ and use IBNR and RBNS factors on an aggregate level to correct for the fact that typically $X \neq X^t$. Note that the information that one uses for reserving is then non-monotone, since for $0\leq s \leq t$, it holds that $X^s_{[0,s]}$ is generally unknown from $X^t_{[0,t]}$ and vice versa. Reserves in the presence of non-monotone information have been studied in~\citet{ChristiansenFurrer2021}. In this, stochastic Thiele differential equations for prospective reserves are derived subject to information deletions, i.e.\ non-monotone information. These might be useful for studying the properties of the reserves $V^t$ currently used in practice.
\demormk
\end{remark}

\begin{example} (Disability insurance with different origins: Reserving.) \label{ex:SimpleReserving} \vspace{0.15cm} \\
We here describe reserving in the transaction time model from Example~\ref{ex:SimpleTransactionTime} and find explicit expressions for $\mathcal{V}$. To obtain intuitive formulas, we impose additional structure on the model for $\mathcal{Z}$ in the form of a conditional independence assumption.

Write $V_j(t,u) = \mathbb{E}[P(t) \mid X_t = (j,u) ]$ for the state-wise valid time reserves in the semi-Markov setup. These may be calculated using known methods, see e.g.~\citet{Christiansen2012} and~\citet{BuchardtMollerSchmidt2015}. Note that on $(Z_t =a)$, we have $\mathcal{F}^X_t = \mathcal{F}^\mathcal{Z}_t$ and $\mathcal{P}(t) = P(t)$ and thus $\mathcal{V}(t)= V_a(t,t)$. On $(Z_t \in \{r,d\})$, we have $\mathcal{P}(t)=P(t)=0$ and thus $\mathcal{V}(t)=0$. Hence only the case $(Z_t \in \{i_1,i_2\} )$ corresponding to an RBNS claim requires consideration. Using Theorem~\ref{theorem:PZequalsPYR}, we get
\begin{align*}
    \mathcal{V}(t) &= \mathbb{E} \bigg[P(t) + \kappa(t) \int_{[\theta_i,t]} \frac{1}{\kappa(s)} (b_{Y_s} - b_{Z_t}) \diff s \, \bigg| \, \mathcal{F}^\mathcal{Z}_t \bigg] \\
    &=
    \frac{\kappa(t)}{\kappa(\theta_i)} \mathbb{E} \big[P(\theta_i) \mid \mathcal{F}^\mathcal{Z}_t \big] - \int_{(\theta_i,t]} \frac{\kappa(t)}{\kappa(s)} b_{Z_t} \diff s.
\end{align*}
Assume that
\begin{align*}
X_{[0,\infty)} \indep \mathcal{F}^\mathcal{Z}_t \mid \mathcal{F}^X_t.
\end{align*}
In other words, if one had known the true value of the biometric state process, additional transaction time information is superfluous. Then by the law of iterated expectations and the semi-Markov property,
\begin{align*}
    \mathbb{E} \big[P(\theta_i) \mid \mathcal{F}^\mathcal{Z}_t \big] = \sum_{k=1}^2 \mathbb{P}(Y_{\theta_i} = i_k \mid \mathcal{F}^\mathcal{Z}_t) \bigg( \frac{\kappa(\theta_i)}{\kappa(t)} V_{i_k}(t,t-\theta_i) + \int_{(\theta_i,t]} \frac{\kappa(\theta_i)}{\kappa(s)} b_{i_k} \diff s \bigg).
\end{align*}
To conclude, on $(Z_t \in \{i_1,i_2\} )$ it holds that
\begin{align} \label{eq:DisabReserveExample}
    \mathcal{V}(t) = \sum_{k=1}^2 \mathbb{P}(Y_{\theta_i} = i_k \mid \mathcal{F}^\mathcal{Z}_t) \bigg( \ V_{i_k}(t,t-\theta_i) + \int_{(\theta_i,t]} \frac{\kappa(t)}{\kappa(s)} (b_{i_k}-b_{Z_t}) \diff s \bigg),
\end{align}
which is an explicit expression for the RBNS reserve. The probabilities $\mathbb{P}(Y_{\theta_i} = i_k \mid \mathcal{F}^\mathcal{Z}_t)$ for $k\in\{1,2\}$ may be calculated as absorption probabilities by extending the state space of the transaction time process to include separate reactivated and dead states for each of the disability origins.

As noted in Remark~\ref{remark:PresentValueValidAndTransactionTime}, the transaction time reserves differ from the valid time reserves through both the present values and the conditioning information. If the payment rates in the two disabled states, that is $b_{i_1}$ and $b_{i_2}$, are equal, then the present values are equal, i.e.\ $\mathcal{P}(t)=P(t)$, and we obtain on $(Z_t \in \{i_1,i_2\} )$ that
\begin{align*}
    \mathcal{V}(t) = \sum_{k=1}^2 \mathbb{P}(Y_{\theta_i} = i_k \mid \mathcal{F}^\mathcal{Z}_t) V_{i_k}(t,t-\theta_i).
\end{align*}
 If we additionally assume that the transition rates from states $i_1$ and $i_2$ are equal, then the difference due to differing conditioning information also disappears, and we get
\begin{align*}
    \mathcal{V}(t) = V_{i_1}(t,t-\theta_i) = V(t),
\end{align*}
meaning that the valid and transaction time reserves agree. 

Note that it is easy to extend the model to include $n$ disabled states $i_1,\ldots,i_n$. We may also extend the example to allow for transition between the disabled states in the valid time model. The disabled states could then represent more diverse and complex phenomena such as the degree of lost earning capacity or diagnoses. Such a model is depicted in Figure~\ref{fig:ExtendedValidTimeModel}. 
\begin{figure}[H]
	\centering
	\scalebox{0.75}{
	      \begin{tikzpicture}[node distance=4em, auto]
	\node[punkt] (g) {active};;
	\node[punktl, right=of g] (i1) {disabled $1$};
    \node[below=1em of i1] (i2) { $\vdots$ };
    \node[punktl, below=1.25em of i2] (in) { disabled $n$ };
    \node[punktl, right=of i1] (r) {reactivated};
	\node[punkt, below =3.5em of in] (dead) {dead}; 
    \node[below =0.5em of in] (dummysouth) {}; 
     \node[right=0.5em of i1] (dummyeast) {};
     \node[left=0.5em of i1] (dummywest) {}; 
     \node[anchor=north east, at=(g.north east)] {$a$};
     \node[anchor=north east, at=(i1.north east)] {$i_1$};
     \node[anchor=north east, at=(in.north east)] {$i_n$};
     \node[anchor=north east, at=(r.north east)] {$r$};
     \node[anchor=north east, at=(dead.north east)] {$d$};
     \node[right=0pt] at ($(i1.north east)+(0.05,0.2)$) {$i$};
	\path 
		($(g.south east)$) edge [pil, bend right = 40] ($(dead.west)$)
        ($(g.east)$) edge [pildotted, thick] ($(dummywest.west)$)
        ($(dummyeast.east)$) edge [pildotted, thick] ($(r.west)$)
        ($(r.south west)$) edge [pil, bend left = 40] ($(dead.east)$)
        ($(dummysouth.south)$) edge [pildotted, thick] ($(dead.north)$)
        ($(i1.south east)$) edge [pil, bend left = 13] ($(in.north east)$)
        ($(in.north west)$) edge [pil, bend left = 13] ($(i1.south west)$)
        ($(i1.south)$) edge [pil, bend left = 25] ($(i2)$)
        ($(i2)$) edge [pil, bend left = 25] ($(i1.south)$)
        ($(i2)+(0,-0.15)$) edge [pil, bend left = 25] ($(in.north)$)
        ($(in.north)$) edge [pil, bend left = 25] ($(i2)+(0,-0.15)$)
	; %+(0.1,0)
    \draw[thick, dotted] ($(i1.north west)+(-0.5,0.5)$) rectangle ($(in.south east)+(0.5,-0.5)$);
 
    \end{tikzpicture}
    }
    \caption{Valid time model from Figure~\ref{fig:simpleValidModel} extended to $n$ disabled states and allowing for transition between the disabled states. To reduce clutter, transitions to and from the disabled states $\{i_1,\ldots,i_n\}$ are represented with a single dotted arrow. } 
    \label{fig:ExtendedValidTimeModel}
\end{figure}
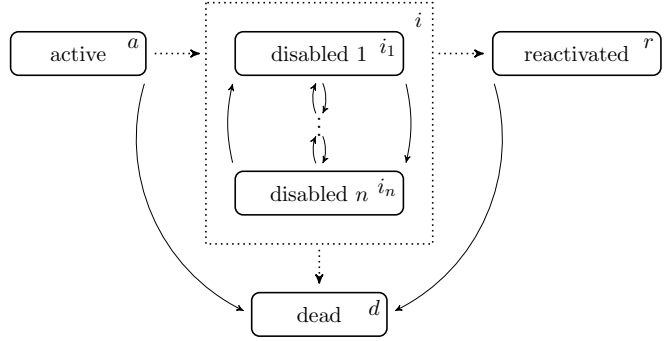
\noindent Allowing for transition between the disabled states comes at a cost in the form of increased complexity of the transaction time model, since the latter needs to be able to generate the valid time process. One option is to let the $Z$-component take values in the state space depicted in Figure~\ref{fig:ComplexTransactionTimeModel}.
\begin{figure}[H]
	\centering
	\scalebox{0.75}{
	   \begin{tikzpicture}[node distance=4em, auto]
	\node[punkt] (g) {active};;
	\node[punktll, right=of g] (i1) {event type $1$};
    \node[below=1em of i1] (i2) { $\vdots$ };
    \node[punktll, below=1.25em of i2] (in) { event type $m$ };
    \node[punktl, right=of i1] (r) {reactivated};
	\node[punkt, below =3.5em of in] (dead) {dead}; 
    \node[below =0.5em of in] (dummysouth) {}; 
     \node[right=0.5em of i1] (dummyeast) {};
     \node[left=0.5em of i1] (dummywest) {}; 
     \node[anchor=north east, at=(g.north east)] {$a$};
     \node[anchor=north east, at=(i1.north east)] {$e_1$};
     \node[anchor=north east, at=(in.north east)] {$e_m$};
     \node[anchor=north east, at=(r.north east)] {$r$};
     \node[anchor=north east, at=(dead.north east)] {$d$};
      \node[right=0pt] at ($(i1.north east)+(0.05,0.2)$) {$e$};
	\path 
		($(g.south east)$) edge [pil, bend right = 40] ($(dead.west)$)
        ($(g.east)$) edge [pildotted, thick] ($(dummywest.west)$)
        ($(dummyeast.east)$) edge [pildotted, thick] ($(r.west)$)
        ($(r.south west)$) edge [pil, bend left = 40] ($(dead.east)$)
        ($(dummysouth.south)$) edge [pildotted, thick] ($(dead.north)$)
        ($(i1.south east)$) edge [pil, bend left = 13] ($(in.north east)$)
        ($(in.north west)$) edge [pil, bend left = 13] ($(i1.south west)$)
        ($(i1.south)$) edge [pil, bend left = 25] ($(i2)$)
        ($(i2)$) edge [pil, bend left = 25] ($(i1.south)$)
        ($(i2)+(0,-0.15)$) edge [pil, bend left = 25] ($(in.north)$)
        ($(in.north)$) edge [pil, bend left = 25] ($(i2)+(0,-0.15)$)
	; %+(0.1,0)
    \draw[thick, dotted] ($(i1.north west)+(-0.5,0.5)$) rectangle ($(in.south east)+(0.5,-0.5)$);
 
    \end{tikzpicture}
    }
	\caption{The $Z$-component of a transaction time process $\mathcal{Z}$ that can generate the extended valid time model of Figure~\ref{fig:ExtendedValidTimeModel}. To reduce clutter, transitions to and from the RBNS states $\{e_1,\ldots,e_m\}$ are represented with a single dotted arrow.}
	\label{fig:ComplexTransactionTimeModel}
\end{figure}
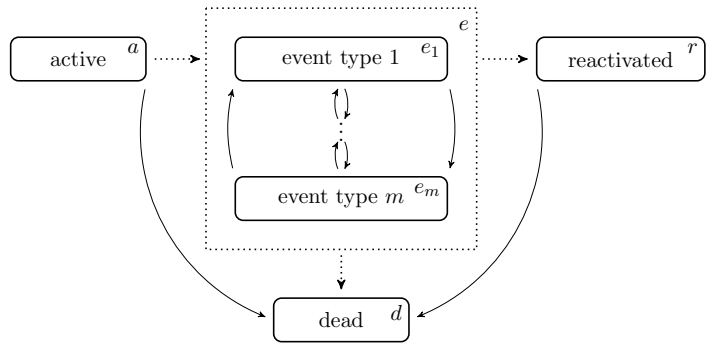
\noindent As noted in Remark~\ref{remark:PresentValueValidAndTransactionTime}, we do not need to specify the full distribution of changes in $H^t_s$ at jumps of $Z$ between states $e_1,\ldots,e_m$. This distribution only affects the reserves through the distribution of the valid time process conditional on the observed information, and it therefore suffices to model this. On $(Z_t \in \{e_1,\ldots,e_m\} )$, corresponding to an RBNS claim, one may show that
\begin{align*}
    \mathcal{V}(t) =\sum_{k=1}^n \bigg(\int_{[0,t]} V_{i_k}(t,u) \mathbb{P}(Y_t = i_k , U_t \in \diff u \mid \mathcal{F}^\mathcal{Z}_t) + \int_{(\theta_i,t]} \frac{\kappa(t)}{\kappa(s)} \mathbb{P}(Y_s = i_k \mid \mathcal{F}^\mathcal{Z}_t) (b_{i_k} - b_{Y^t_s} ) \diff s \bigg).
\end{align*}
This identity is comparable to~\eqref{eq:DisabReserveExample}, but it is more complex and requires one to model the path the insured takes through the disabled states between time $\theta_i$ and time $t$. This may be seen from the dependence on $U_t$ in the first term and the dependence on $Y_s$ for $s \in (\theta_i,t]$ in the second term.   

Finally, it is easy to extend the example to allow for general semi-Markov payments of the form described in Example~\ref{ex:MarkovSemiMarkovCashFlow}. This includes risk periods, waiting periods, and transition payments.\demoex
\end{example}

\noindent The above example serves as a simple theoretical demonstration of the potential of our general framework. As there is no reporting delay, there is no IBNR reserve in Example~\ref{ex:SimpleReserving}. The RBNS reserve of the example also only differs from the valid time disability reserve due to the imperfect observation of the disability type. To capture the full picture of IBNR and RBNS reserving, one would need to explore more intricate transaction time models with both reporting delays and claim adjudications. While this extension is outside the scope of this paper, our general framework readily allows for such continued studies.  We stress that such applications are the main motivation for introducing the transaction time framework. In addition to developing model extensions, it could also be relevant to develop estimation procedures for the above example as well as for more complicated models.

\subsubsection*{Reserve dynamics}

The study of reserve dynamics is of great importance, especially in relation to model validation, Cantelli's theorem and reserve-dependent payments, Hattendorff's theorem on non-correlation between losses, and the emergence and decomposition of surplus as well as sensitivity analyses, cf.~Section~1 in~\citet{ChristiansenFurrer2021}. We conclude this section by deriving the dynamics of the transaction time reserve $\mathcal{V}$ and valid time reserve $V$ following the same procedure as for the classic reserve, see e.g.~\citet{Christiansen2020}. Essentially, this amounts to applying an explicit martingale representation theorem to $\mathcal{V}$; the idea of applying martingale representation techniques dates back to~\citet{Norberg1992}. The dynamics of the prospective reserves bears a resemblance to Thiele’s differential equation; one might even say it constitutes a stochastic version of Thiele's differential equation. The literature, however, seems to reserve the term \textit{stochastic Thiele equation} for the stochastic differential equation related to the so-called state-wise prospective reserves, see e.g.~\citet{ChristiansenFurrer2021}.

Recall the definitions $\mathcal{V}(t) = \mathbb{E}\big[\mathcal{P}(t) \mid \mathcal{F}^\mathcal{Z}_t\big]$ and $\mathcal{Z}_t = (Z_t,H^t_t)$. Define a random counting measure $\mu_{\mathcal{Z}}$ corresponding to $\mathcal{Z}$ by 
\begin{align*}
    \mu_{\mathcal{Z}}(C) &= \sum_{n=1}^\infty 1_C(T_n,\mathcal{Z}_{T_n}), \hspace{0.5cm} C \in \mathbb{B}( [0,\infty) )  \otimes \mathbb{B}(\mathbb{R}^\infty),
\end{align*}
and let $\Lambda_{\mathcal{Z}}$ be its \textit{compensating measure}, given in Definition 4.3.2 (iii) of~\citet{Jacobsen2006}. By Theorem 4.5.2 of~\citet{Jacobsen2006}, if $\mathbb{E}[\mu_{\mathcal{Z}}([0,t] \times D) ] < \infty$ for all $t \geq 0$ and $D \in \mathbb{B}(\mathbb{R}^\infty)$, we have that 
\begin{align*}
t \mapsto \mu_{\mathcal{Z}}([0,t] \times D) - \Lambda_{\mathcal{Z}}([0,t] \times D)
\end{align*}
is a martingale for any $D \in \mathbb{B}(\mathbb{R}^\infty)$. Let $\xi_n = (T_1,...,T_n ; \mathcal{Z}_{T_1},...,\mathcal{Z}_{T_n} )$ be the MPP history of $\mathcal{Z}$ at time $T_n$.

Write $\zeta=(\zeta_z,\zeta_h)$ for a generic realization of $\mathcal{Z}_t$, where the coordinates $\zeta_z$ and $\zeta_h$ pertain to $Z_{t}$ and $H^{t}_{t}$, respectively. Finally, define the sums at risk in the transaction time model for a jump of $\mathcal{Z}$ to $\zeta$ at time $t$:
\begin{align*}
    \mathcal{R}(t,\zeta) &= \sum_{n=1}^\infty 1_{(T_n < t \leq T_{n+1})} \Big( \kappa(t)\Big(B^\circ( (f_{ \zeta_h }(s))_{0 \leq s \leq t} ) -B^\circ( X^{t-}_{[0,t]} ) \Big) \\
    & \qquad \qquad \qquad \qquad \qquad +\mathbb{E}[ \mathcal{P}(t) \mid \xi_n, (T_{n+1},\mathcal{Z}_{T_{n+1}})=(t,\zeta)] -\mathbb{E}[\mathcal{P}(t) \mid \xi_n, T_{n+1} > t] \Big).
\end{align*}
This is a difference in payments and reserves at time~$t$ between a jump-to-$\zeta$ and a remain-in-$\mathcal{Z}_{t-}$ scenario.
\begin{remark} (Definition of non-standard conditional expectations.) \label{remark:Conditioning} \vspace{0.15cm} \\
One should be careful about the definition of $\mathbb{E}[\mathcal{P}(t) \mid \xi_n, T_{n+1} > t]$ and similar quantities outside $(T_{n+1} > t)$, confer with e.g.\ \citet{ChristiansenFurrer2021}. In this paper, it corresponds to the version 
\begin{align*}
    \mathbb{E}[\mathcal{P}(t) \mid \xi_n, T_{n+1} > t] = \frac{\mathbb{E}[\mathcal{P}(t) 1_{(T_{n+1} > t)} \mid \xi_n]}{\mathbb{E}[1_{(T_{n+1} > t)} \mid \xi_n]}
\end{align*}
under the convention $0/0=0$ and where the expectations are the regular conditional expectations constructed in~\citet{Jacobsen2006}. That this version is the relevant one follows from the proof of Theorem~\ref{thm:Dynamics}.
\demormk
\end{remark}

\noindent We then have the following theorem:

\begin{theorem} (Transaction time reserve dynamics.) \label{thm:Dynamics} \vspace{0.15cm} \\
For $t \geq 0$, it holds that

\begin{align}\label{eq:transac_dyn}
    \mathcal{V}(\diff t) &= \mathcal{V}(t-) \frac{\kappa(\diff t)}{\kappa(t-)}-\mathcal{B}(\diff t) +  \int_{\mathbb{R}^\infty}  \mathcal{R}(t,\zeta) \;  (\mu_{\mathcal{Z}}-\Lambda_{\mathcal{Z}})(\diff t,\diff \zeta). 
\end{align}
\end{theorem}

\begin{proof}
Introduce 
\begin{align*}
\mathcal{P}^{\circ}(t) = \frac{1}{\kappa(t)}\mathcal{P}(t).
\end{align*}
We have that $\mathcal{P}^{\circ}(0) = \mathcal{P}(0)$, which is assumed integrable, so we can define 
\begin{align*}
t \mapsto M_t = \mathbb{E}\big[\mathcal{P}^{\circ}(0) \mid \mathcal{F}^\mathcal{Z}_t\big],
\end{align*}
which is a martingale. Since $\mathcal{V}$ is presumed $\mathcal{F}^{\mathcal{Z}}$-adapted and $\mathbb{P}$-a.s.\ right-continuous, the same holds for a version of $M$, since $M_t=\frac{1}{\kappa(t)}\mathcal{V}(t)+\mathcal{P}^{\circ}(0)-\mathcal{P}^{\circ}(t)$. Then a martingale representation theorem, namely Theorem 4.6.1 of~\citet{Jacobsen2006}, gives the existence of predictable processes $S^\zeta_s$ such that
\begin{align*}
M_t = M_0 + \int_{(0,t] \times \mathbb{R}^\infty} S^\zeta_s \; (\mu_{\mathcal{Z}}-\Lambda_{\mathcal{Z}})(\diff s,\diff \zeta)
\end{align*}
$\mathbb{P}$-a.s.\ simultaneously over $t$. 
Using the adaptedness of $M$, we can, as in the proof of the aforementioned Theorem~4.6.1, use Proposition~4.2.1(biii) of~\citet{Jacobsen2006} to write
\begin{align*}
    M_t &= \sum_{n=0}^\infty 1_{(T_n \leq t < T_{n+1})} g^{n}_{\xi_n}(t)
\end{align*}
for measurable functions $(h_n,t) \mapsto g^n_{h_n}(t)$. Due to $M$ being a conditional expectation, we can use Corollary 4.2.2 of~\citet{Jacobsen2006} to identify
\begin{align*}
g^{n}_{\xi_n}(t) = \mathbb{E}[\mathcal{P}^{\circ}(0) \mid \xi_n, T_{n+1} > t]
\end{align*}
on $(T_n \leq t < T_{n+1})$. To identify the function for all $(h_n,t)$, we first observe that according to Remark~4.2.3 of~\citet{Jacobsen2006},
\begin{align*}
    g^{n}_{\xi_n}(t) = \frac{\mathbb{E}[\mathcal{P}^{\circ}(0) 1_{(T_{n+1} > t)} \mid \xi_n]}{\mathbb{E}[1_{(T_{n+1} > t)} \mid \xi_n]}
\end{align*}
on $(T_n \leq t < T_{n+1})$. Define the functions $(h_n,t) \mapsto g^n_{h_n}(t)$ as
\begin{align*}
    g^n_{h_n}(t) = \frac{\mathbb{E}[\mathcal{P}^{\circ}(0) 1_{(T_{n+1} > t)} \mid \xi_n = h_n]}{\mathbb{E}[1_{(T_{n+1} > t)} \mid \xi_n = h_n]}
\end{align*}
on the set $D_n=\{(h_n,t): \mathbb{E}[1_{(T_{n+1} > t)} \mid \xi_n = h_n] \neq 0 \}$ and zero otherwise. These functions are well-defined since the conditional expectations are regular and fixed. By the above calculations, they satisfy the required identity of Proposition~4.2.1(biii) in~\citet{Jacobsen2006}. For the measurability condition, note first that $\mathbb{E}[1_{(T_{n+1} > t)} \mid \xi_n = h_n]$ is measurable as a function of $h_n$ since it is a regular conditional expectation, and that it is jointly measurable as a function of $(h_n,t)$ since it is right-continuous as a function of $t$ for any $h_n$ by the dominated convergence theorem. This implies that $D_n$ is measurable. By the same arguments, $\mathbb{E}[\mathcal{P}^{\circ}(0) 1_{(T_{n+1} > t)} \mid \xi_n = h_n]$ is seen to be jointly measurable as a function of $(h_n,t)$. From this we may conclude that $(h_n,t) \mapsto g^n_{h_n}(t)$ is measurable, so it especially satisfies Proposition~4.2.1(biii) in~\citet{Jacobsen2006}. In the following, we write $\mathbb{E}[\mathcal{P}^{\circ}(0) \mid \xi_n, T_{n+1} > t]$ for $g^n_{\xi_n}(t)$, but this is merely notation; calculations with $\mathbb{E}[\mathcal{P}^{\circ}(0) \mid \xi_n, T_{n+1} > t]$ actually use the properties of $g_{h_n}^n(t)$.

The proof of the aforementioned Theorem 4.6.1 furthermore gives that 
\begin{align*}
    S^\zeta_t &= \sum_{n=0}^\infty 1_{(T_n < t \leq T_{n+1})} \Big( g^{n+1}_{(\xi_n,(t,\zeta))}(t)-g^{n}_{\xi_n}(t) \Big),
\end{align*}
so that
\begin{align*}
    S^\zeta_t &= \sum_{n=0}^\infty 1_{(T_n < t \leq T_{n+1})} \big( \mathbb{E}[\mathcal{P}^{\circ}(0) \mid \xi_n, (T_{n+1},\mathcal{Z}_{T_{n+1}})=(t,\zeta)] -  \mathbb{E}[\mathcal{P}^{\circ}(0) \mid \xi_n, T_{n+1} > t] \big) \\
    &= \sum_{n=0}^\infty 1_{(T_n < t \leq T_{n+1})} \big( \mathbb{E}[\mathcal{P}^{\circ}(t-) \mid \xi_n, (T_{n+1},\mathcal{Z}_{T_{n+1}})=(t,\zeta)] -  \mathbb{E}[\mathcal{P}^{\circ}(t-) \mid \xi_n, T_{n+1} > t] \big)
\end{align*}
using that $\mathcal{P}^{\circ}(0)-\mathcal{P}^{\circ}(t-) = \int_{(0,t)} \frac{1}{\kappa(s)} \mathcal{B}(ds)$ are $\xi_n$-measurable on $(T_n < t \leq T_{n+1} )$. Therefore, the dynamics of $M$ is
\begin{align*}
    \diff M_t &= \int_{\mathbb{R}^\infty} S^\zeta_t \: (\mu_{\mathcal{Z}} - \Lambda_{\mathcal{Z}})(\diff t,\diff \zeta) \\
    &= \sum_{n=0}^\infty \int_{\mathbb{R}^\infty} 1_{(T_n < t \leq T_{n+1})} \Big( \mathbb{E}[\mathcal{P}^{\circ}(t-) \mid \xi_n, (T_{n+1},\mathcal{Z}_{T_{n+1}})=(t,\zeta)] \\
    & \qquad \qquad \qquad \qquad \qquad \quad  -  \mathbb{E}[\mathcal{P}^{\circ}(t-) \mid \xi_n, T_{n+1} > t] \Big) \: (\mu_{\mathcal{Z}} - \Lambda_{\mathcal{Z}})(\diff t,\diff \zeta).
\end{align*}
Using that 
\begin{align} \label{eq:PZ'}
    \mathcal{P}^{\circ}(t)-\mathcal{P}^{\circ}(0) = -\int_{(0,t]} \frac{1}{\kappa(s)} \mathcal{B}(\diff s)
\end{align}
is $\mathcal{F}^\mathcal{Z}$-adapted, we get
\begin{align*}
\mathbb{E}[\mathcal{P}^{\circ}(0) \mid \mathcal{F}^\mathcal{Z}_t] - \mathbb{E}[\mathcal{P}^{\circ}(0) \mid \mathcal{F}^\mathcal{Z}_0] = \mathbb{E}[\mathcal{P}^{\circ}(t) \mid \mathcal{F}^\mathcal{Z}_t] - \mathbb{E}[\mathcal{P}^{\circ}(0) \mid \mathcal{F}^\mathcal{Z}_0] - (\mathcal{P}^{\circ}(t)-\mathcal{P}^{\circ}(0)),
\end{align*}
which upon rearrangement becomes
\begin{align*}
\mathbb{E}[\mathcal{P}^{\circ}(t) \mid \mathcal{F}^\mathcal{Z}_t] - \mathbb{E}[\mathcal{P}^{\circ}(0) \mid \mathcal{F}^\mathcal{Z}_0] = \mathcal{P}^{\circ}(t)-\mathcal{P}^{\circ}(0) + \mathbb{E}[\mathcal{P}^{\circ}(0) \mid \mathcal{F}^\mathcal{Z}_t] - \mathbb{E}[\mathcal{P}^{\circ}(0) \mid \mathcal{F}^\mathcal{Z}_0].
\end{align*}
Introducing
\begin{align*}
\mathcal{V}^{\circ}(t) = \frac{1}{\kappa(t)}\mathcal{V}(t) = \mathbb{E}[\mathcal{P}^{\circ}(t) \mid \mathcal{F}^\mathcal{Z}_t],
\end{align*}
we can write this as 
\begin{align*}
\mathcal{V}^{\circ}(t)-\mathcal{V}^{\circ}(0) = \mathcal{P}^{\circ}(t)-\mathcal{P}^{\circ}(0) + M_t - M_0.
\end{align*}
The identity~\eqref{eq:PZ'} furthermore gives
\begin{align*}
\mathcal{P}^{\circ}(\diff t) = -\frac{1}{\kappa(t)}\mathcal{B}(\diff t).
\end{align*}
The above calculations imply
\begin{align*}
    \mathcal{V}^{\circ}(\diff t) &= \mathcal{P}^{\circ}(\diff t) + \diff M_t \\
    &= -\frac{1}{\kappa(t)} \mathcal{B}(\diff t) + \sum_{n=0}^\infty \int_{\mathbb{R}^\infty} 1_{(T_n < t \leq T_{n+1})} \Big( \mathbb{E}[\mathcal{P}^{\circ}(t-) \mid \xi_n, (T_{n+1},\mathcal{Z}_{T_{n+1}})=(t,\zeta)] \\
    & \qquad \qquad \qquad \qquad \qquad \qquad \qquad \qquad \quad -  \mathbb{E}[\mathcal{P}^{\circ}(t-) \mid \xi_n, T_{n+1} > t] \Big) \: (\mu_{\mathcal{Z}} - \Lambda_{\mathcal{Z}})(\diff t,\diff \zeta).
\end{align*}
The time $t$ payment $\frac{1}{\kappa(t)} \mathcal{B}(\{t\})$ can be taken out of both intergrands, and this amounts to $B^\circ( (f_{\zeta_h }(s))_{0 \leq s \leq t} ) -B^\circ( X^{t-}_{[0,t]} )$. It is the difference in the payment at time $t$ between a jump and a remain scenario when $\mathcal{Z}$ jumps to $\zeta$. Taking out the time $t$ payment, we get $\mathcal{P}^{\circ}(t-) = \frac{1}{\kappa(t)}\mathcal{B}(\{t\})+\mathcal{P}^{\circ}(t)$, and using integration by parts, we finally have
\begin{align*}
    \mathcal{V}(\diff t) &= \diff (\kappa(t)\mathcal{V}^{\circ}(t)) \\
    &= \mathcal{V}^{\circ}(t-) \kappa(\diff t)+\kappa(t) \mathcal{V}^{\circ}(\diff t) \\
    &= \mathcal{V}(t-) \frac{\kappa(\diff t)}{\kappa(t-)}-\mathcal{B}(\diff t) \\
    & \quad + \sum_{n=1}^\infty \int_{\mathbb{R}^\infty}  1_{(T_n < t \leq T_{n+1})} \Big( B^\circ( (f_{\zeta_h }(s))_{0 \leq s \leq t} ) -B^\circ( X^{t-}_{[0,t]} )  \\
    & \qquad \qquad +\mathbb{E}[\mathcal{P}(t) \mid \xi_n, (T_{n+1},\mathcal{Z}_{T_{n+1}})=(t,\zeta)] -\mathbb{E}[\mathcal{P}(t) \mid \xi_n, T_{n+1} > t] \Big) \, (\mu_{\mathcal{Z}}-\Lambda_{\mathcal{Z}})(\diff t,\diff \zeta),
\end{align*}
which yields the desired result by definition of the sums at risk.
\end{proof}
\noindent Theorem~\ref{thm:Dynamics} shows that the transaction time reserve $\mathcal{V}$ changes with interest accrual $\mathcal{V}(t-) \frac{\kappa(\diff t)}{\kappa(t-)}$, actual benefits less premiums $\mathcal{B}(\diff t)$ and a martingale term $\int_{\mathbb{R}^\infty}  \mathcal{R}(t,\zeta) \;  (\mu_{\mathcal{Z}}-\Lambda_{\mathcal{Z}})(\diff t,\diff \zeta)$, which is the sums at risk integrated with respect to the underlying compensated random counting measure. The martingale term may be interpreted as stochastic noise since it is a mean-zero process, and may thus be used for model validation and back-testing purposes. Actual applications are outside the scope of this paper.

One could alternatively have derived Theorem~\ref{thm:Dynamics} from Theorem~7.1 in~\citet{Christiansen2021Time}, which is an explicit martingale representation theorem that holds even when the information being conditioned on is non-monotone. The proof presented here is however more concise, as our information $\mathcal{F}^\mathcal{Z}$ is monotone, so more standard results apply. Theorem~\ref{thm:Dynamics} is similar to Proposition~3.2 in~\citet{Christiansen2020}, but differs among other things by not being restricted to state processes taking values in a finite space.

\begin{remark} (Dynamics of valid time reserve.) \label{remark:ClassicDynamics} \vspace{0.15cm} \\
Define the random counting measure $\mu_{X}$ corresponding to $X$ by 
\begin{align*}
    \mu_X(C) &= \sum_{n=1}^\infty 1_C(\tau_n,X_{\tau_n}), \hspace{0.5cm} C \in \mathbb{B}( [0,\infty) )  \otimes \mathbb{B}(\mathbb{R}^d).
\end{align*} Let $\Lambda_X$ be the compensating measure for $\mu_X$, and let $\gamma_n = (\tau_1,...,\tau_n ; X_{\tau_1},...,X_{\tau_n})$ be the MPP history of $X$ at time~$\tau_n$. By the same calculations as for Theorem~\ref{thm:Dynamics}, we find the dynamics of the valid time reserve $V$:
\begin{align}\label{eq:valid_dyn}
    V(\diff t) &= V(t-) \frac{\kappa(\diff t)}{\kappa(t-)}-B(\diff t) +  \int_{\mathbb{R}^d}  R(t,y) \;  (\mu_{X}-\Lambda_{X})(\diff t,\diff y) 
\end{align}
for the sums at risk 
\begin{align*}
    R(t,y) &= \sum_{n=1}^\infty 1_{(\tau_n < t \leq \tau_{n+1})} \Big(B\big( (f_{ (H_{t-},(t,y)) }(s))_{0 \leq s \leq t} ,\{t\}\big) - B\big( (f_{ H_{t-} }(s))_{0 \leq s \leq t}  ,\{t\}\big) \\ & \quad \qquad \qquad  \qquad \qquad  +\mathbb{E}[ P(t) \mid \gamma_n, (\tau_{n+1},X_{\tau_{n+1}})=(t,y)] -\mathbb{E}[P(t) \mid \gamma_n, \tau_{n+1} > t] \Big).
\end{align*}
This result is again similar to Proposition~3.2 in~\citet{Christiansen2020}, but still differs among other things by not being restricted to state processes taking values in a finite space. The conditional expectations are to be interpreted as in Remark~\ref{remark:Conditioning}.

Suppose now that $X$ is a pure Markov jump process on a finite state space $E=\{1,2,...,J\}$ with payments specified as in Example~\ref{ex:MarkovSemiMarkovCashFlow}. In other words, the valid time payments consist of deterministic sojourn payments $t \mapsto B_j(t)$ and deterministic transition payments $t \mapsto b_{jk}(t)$. Then
\begin{align*}
B\big( (f_{ (H_{t-},(t,y)) }(s))_{0 \leq s \leq t} ,\{t\}\big) - B\big( (f_{ H_{t-} }(s))_{0 \leq s \leq t}  ,\{t\}\big) = b_{X_{t-}y}(t).
\end{align*}
Furthermore,
\begin{align*}
\mu_X(\mathrm{d}t, \{k\})
=
N_{X_{t-}k}(\mathrm{d}t)
\end{align*}
and, since $X$ is Markovian,
\begin{align*}
\Lambda_X(\mathrm{d}t, \{k\})
=
\Lambda_{X_{t-}k}(\mathrm {d}t) 
\end{align*}
for suitably regular cumulative transition rates $t \mapsto \Lambda_{jk}(t)$.  Consequently, by invoking the Markov property, the dynamics~\eqref{eq:valid_dyn} read
\begin{align}\label{eq:valid_dyn_Markov}
\begin{split}
    V(\diff t) &= V(t-) \frac{\kappa(\diff t)}{\kappa(t-)}-B(\diff t) \\
    &\quad + \sum_{ \substack{j,k = 1 \\ j \neq k } }^J \mathds{1}_{(X_{t-} = j)} \big(b_{jk}(t) +\mathbb{E}[P(t) \, | \, X_t = k] - \mathbb{E}[P(t) \, | \, X_t = j]\big) \big(N_{jk}(\mathrm{d}t) - \Lambda_{jk}(\mathrm{d}t)\big).
\end{split}
\end{align}
This constitutes a significant simplification.\demormk
\end{remark}

\noindent In comparing~\eqref{eq:transac_dyn} with~\eqref{eq:valid_dyn}, it is apparent that the transaction and valid time reserves admit comparable dynamics. In both cases, there is a contribution due to interest accrual, a contribution from benefits less premiums, and finally a martingale term. In general, the dynamics of the transaction time reserve are more complicated than that of the valid time reserve -- for two reasons. First, the martingale term is more involved, which stems from the fact that the model for $\mathcal{Z}$ is typically more elaborate than that for $X$. Second, the accumulated cash flow in transaction time $\mathcal{B}$ is a complicated function of, among other things, the accumulated cash flow in valid time $B$. The difference might be particularly striking under the quite common assumption that $X$ is a pure Markov jump process on a finite state space $E = \{1,2,\ldots,J\}$ and the valid time payments consist of deterministic sojourn and transition payments. In this case, the dynamics of the valid time reserve simplify, cf.~\eqref{eq:valid_dyn_Markov}, but there is in general no reason why this simplification should carry over to the transaction time reserve -- unless further assumptions are imposed.

\section*{Acknowledgments and declarations of interest}

We would like to thank an anonymous referee for very helpful comments and suggestions. Oliver Lunding Sandqvist's research has partly been funded by the Innovation Fund Denmark (IFD) under File No.\ 1044-00144B. The authors declare no conflicts of interest.

\end{document}